\documentclass{article}
\usepackage{natbib}
\usepackage{amsmath}  %% <- after lineno
\usepackage{etoolbox} %% <- for \cspreto, \csappto
\usepackage{amsfonts}
\usepackage{amssymb}
\usepackage{amsthm}
\usepackage{dsfont}
\usepackage{thmtools, thm-restate}
\declaretheorem{lemma}
\usepackage{graphicx}
\usepackage{xcolor}
\usepackage{caption, subcaption}
\usepackage{float}
\usepackage[hidelinks]{hyperref}
\usepackage{appendix}
\usepackage{booktabs}
\usepackage{numprint}
\usepackage[scientific-notation=true]{siunitx}
\usepackage{array}
\usepackage{xspace}

\usepackage{accents}

\newtheorem{thm}{Theorem}
\newtheorem{cor}{Corollary}[thm]

\newtheorem{defn}{Definition}
\newtheorem{obs}[lemma]{Observation}
\newtheorem{alg}{Algorithm}

\newcommand{\tl}{\triangleleft}

\newcommand{\N}{\mathbb{N}}

\newcommand{\Part}{\operatorname{Part}}
\newcommand{\cu}{\operatorname{CU}}
\newcommand{\tgts}{\operatorname{Ch}}
\newcommand{\shbelow}{\operatorname{B}}
\newcommand{\st}{\,\middle\vert\,}

\newcommand{\sarscov}{SARS-CoV-2\xspace}
\newcommand{\usher}{UShER\xspace}
\newcommand{\dnapars}{\texttt{dnapars}\xspace}
\newcommand{\matoptimize}{\texttt{matOptimize}\xspace}
\newcommand{\phylip}{\texttt{PHYLIP}\xspace}
\newcommand{\gctree}{\texttt{gctree}\xspace}
\newcommand{\historydag}{\texttt{historydag}\xspace}

\definecolor{brewer-pink}{HTML}{e7298a}
\definecolor{brewer-orange}{HTML}{d95f02}
\definecolor{brewer-purple}{HTML}{7570b3}
\definecolor{brewer-green}{HTML}{1b9277}
\definecolor{darkred}{HTML}{C40004}

\newcommand{\slug}{\hbox{\kern1.5pt\vrule width2.5pt height6pt depth1.5pt\kern1.5pt}}
\def\xskip{\hskip 7pt plus 3pt minus 4pt}
\newdimen\algindent
\newif\ifitempar \itempartrue % normally true unless briefly set false
\def\algindentset#1{\setbox0\hbox{{\bf #1.\kern.25em}}\algindent=\wd0\relax}
\def\algbegin #1 #2{\algindentset{#21}\alg #1 #2} % when steps all have 1 digit
\def\aalgbegin #1 #2{\algindentset{#211}\alg #1 #2} % when 10 or more steps
\def\alg#1(#2). {\medbreak % Usage: \algbegin Algorithm A (algname). This...
  \noindent{\bf#1}({\it#2\/}).\xskip\ignorespaces}
\def\kalgstep#1.{\ifitempar\smallskip\noindent\else\itempartrue
   \hskip-\parindent\fi
   \hbox to\algindent{\bf\hfil #1.\kern.25em}%
   \hangindent=\algindent\hangafter=1\ignorespaces}

\newenvironment{taocpalg}[3]{%
\vspace{1em}%
\algbegin Algorithm #1. ({#2}). #3 }
{\vspace{1em}}

\renewcommand{\thefootnote}{\arabic{footnote}}
\newcommand{\symfootnote}[1]{%
\let\oldthefootnote=\thefootnote%
\stepcounter{mpfootnote}%
\addtocounter{footnote}{-1}%
\renewcommand{\thefootnote}{\fnsymbol{mpfootnote}}%
\footnote{#1}%
\let\thefootnote=\oldthefootnote%
}

\title{Representing and extending ensembles of parsimonious evolutionary histories with a directed acyclic graph}
\author{
    Will Dumm$^{1,2}$\\
    Mary Barker$^{1,2}$\\
    William Howard-Snyder$^3$\\
    William S DeWitt III$^4$\\
    Frederick A Matsen IV$^{1,2,5,6,*}$\\
}
\date{October 11, 2023}

\begin{document}
\delimitershortfall=-1pt
\maketitle

\renewcommand*{\thefootnote}{\fnsymbol{footnote}}
 \footnotetext[1]{Correspondence: matsen@fredhutch.org}
\renewcommand*{\thefootnote}{\arabic{footnote}}
 \footnotetext[1]{Computational Biology Program, Fred Hutchinson Cancer Research Center, Seattle, Washington, USA}
 \footnotetext[2]{Howard Hughes Medical Institute, Computational Biology Program, Fred Hutchinson Cancer Research Center, Seattle, Washington, USA}
 \footnotetext[3]{Paul G. Allen School of Computer Science and Engineering, University of Washington, Seattle, Washington, USA}
 \footnotetext[4]{Department of Electrical Engineering and Computer Sciences, University of California, Berkeley, California, USA}
 \footnotetext[5]{Department of Genome Sciences, University of Washington, Seattle, Washington, USA}
 \footnotetext[6]{Department of Statistics, University of Washington, Seattle, Washington, USA}

 %EM I'm 0000-0003-0607-6025
 %WSD 0000-0002-6802-9139
 %WD 0000-0002-8617-476X
 %MB  0000-0002-7829-8017
 %WHS 0000-0001-7375-8961

\begin{abstract}

In many situations, it would be useful to know not just the best phylogenetic tree for a given data set, but the collection of high-quality trees.
This goal is typically addressed using Bayesian techniques, however, current Bayesian methods do not scale to large data sets.
Furthermore, for large data sets with relatively low signal one cannot even store every good tree individually, especially when the trees are required to be bifurcating.
In this paper, we develop a novel object called the ``history subpartition directed acyclic graph'' (or ``history sDAG'' for short) that compactly represents an ensemble of trees with labels (e.g.\ ancestral sequences) mapped onto the internal nodes.
The history sDAG can be built efficiently and can also be efficiently trimmed to only represent maximally parsimonious trees.
We show that the history sDAG allows us to find many additional equally parsimonious trees, extending combinatorially beyond the ensemble used to construct it.
We argue that this object could be useful as the ``skeleton'' of a more complete uncertainty quantification.

% Intensive sampling and genomic sequencing of rapidly evolving systems, such as viral populations and antibody repertoires, often yields samples with identical sequences, and with every sequence at most a handful of mutations away from its nearest neighbor.
% This poses difficulties for the two standard approaches to phylogenetic uncertainty quantification.
% First, a Bayesian approach would assign bifurcating tree topologies and branch lengths to collections of identical sequences, using many parameters in this signal-weak setting.
% Second, with a bootstrap approach the low sequence divergence means that individual mutations can be essential for defining clades, so the support value will be dominated by whether a given individual mutation is included in the bootstrap sample.
% This leads to uniformly low bootstrap support, even for aspects of the tree that are relatively certain.
% In this paper, we work toward phylogenetic uncertainty quantification for intensively sampled systems using a simple initial goal: characterize the ensemble of maximally parsimonious trees.
% To do so, we introduce a novel object called the ``history subpartition directed acyclic graph'' (or ``history sDAG'' for short) that compactly represents such an ensemble of trees.
% We also show that the history sDAG allows us to find many additional equally parsimonious trees, extending beyond the ensemble used to construct it\@.\\
\emph{\textbf{Keywords:} Maximum Parsimony, Phylogenetic Uncertainty, Phylogenetic Inference, Directed Acyclic Graph}\\
\emph{\textbf{Mathematics Subject Classification:} 92B10, 92-08, 92-04}
%MSC codes:
% 	92-08  	Computational methods for problems pertaining to biology
% 	92B10  	Taxonomy, cladistics, statistics in mathematical biology
% 	92-04  	Software, source code, etc. for problems pertaining to biology
\end{abstract}

\section{Introduction}
Here we develop a structure that can compactly represent and extend collections of phylogenetic trees with ancestral sequences mapped on the internal nodes.
One motivation for this structure comes from uncertainty quantification in statistical phylogenetics, which is typically approached via one of two ways.
Bayesian analysis attempts to characterize the posterior distribution of phylogenetic trees given data: the collection of trees that credibly explain the data, and their probabilities of being the generative tree.
On the other hand, the phylogenetic bootstrap~\citep{Felsenstein1985-ei} resamples columns of the multiple sequence alignment, infers an optimal tree for each one of the resampled data sets, then aggregates features of the resulting trees.

Neither of these are tenable for very large and densely sampled data sets, such as for severe acute respiratory syndrome coronavirus 2 (\sarscov) collections.
Traditional Bayesian analysis is often too slow to apply to these large data sets, and introduces many extra unknown model parameters in a signal-weak setting.
Bootstrapping may remain fast enough when using recent approximations~\citep{Hoang2018-pt}, but has a different problem: it is common for well-established clades (supported by other data) to be supported on the sequence level by a single mutation, so the bootstrap support of the corresponding clade will exactly equal the frequency with which we draw that mutation in the bootstrap sample.
Thus, the bootstrap underestimates support in this case~\citep{Wertheim2022-bm}.

Phylogenetic placement offers a different type of uncertainty estimate: an assessment of the level of certainty in inserting a new sequence into an existing phylogeny.
However, these assessments of uncertainty are relative to a fixed reference tree.
For \sarscov this can be done in the \usher framework~\citep{Turakhia2021-yl}, in which this insertion procedure is used for iterative tree building.
No attempt is made to characterize uncertainty of the complete tree in this framework.

The lack of uncertainty quantification may have consequences for interpretation of \sarscov evolution.
For example, the current practice for the PANGO nomenclature system~\citep{Rambaut2020-lj} for \sarscov does not require any sort of support estimation.
A typical workflow involves placement and local tree construction.
If there is indeed high probability of a single tree, then this is fine.
If not, this seems potentially problematic.

We argue as follows that the diversity of maximally parsimonious trees on the data can be used to bound uncertainty from below.
First, if there are more maximally-parsimonious explanations of the data, this decreases the probability that any one explanation is correct.
For this reason, we expect there to be an inverse relationship between the number of maximally-parsimonious explanations of the data and the certainty of a given node or other feature in the tree.
Furthermore, this inverse relationship should express a lower bound on the uncertainty because there are many other potential compelling trees that are not quite maximally parsimonious.
In any case, analyzing even just the maximally parsimonious set of trees commonly involves so many trees that storing them individually and learning from them with existing techniques is computationally prohibitive.
This is especially the case with parsimony analysis of large data sets, such as those for SARS-CoV-2~\citep{Turakhia2021-yl,MatOptimize22Ye}.

As a second motivation for our work, we also suggest that gathering a collection of maximally parsimonious trees could be helpful for Bayesian analysis.
Although the parsimony criterion is of course not the same as likelihood, the two objectives are closely linked in the case where sequences are densely sampled relative to the amount of evolution~\citep{thornlow2021online}.
Previous work has shown how closely related sequences can greatly inflate the posterior distribution~\citep{Whidden2015-eq}, and a parsimony analysis would have revealed this inflation.
Thus, we hope to use the collection of maximally parsimonious trees as an aid for designing proposal distributions, extending previous successful strategies~\citep{Zhang2020-gd}, and for quantifying exploration of tree space.

In this paper, we formalize a data structure called the \emph{history subpartition directed acyclic graph (a.k.a.\ history sDAG)} to characterize the ensemble of maximally parsimonious trees for large data sets.
This is related to the idea of characterizing the trees in a single optimal ``terrace'' in phylogenetic tree space, with respect to parsimony~\citep{terraces2015sanderson, terraces2011sanderson}.
We describe algorithms to build history sDAGs from internally labeled trees, collapse edges with no mutations, and trim history sDAGs to express only trees which are optimal according to general criteria, such as parsimony.
Although history sDAG construction is not the same as uncertainty estimation, which would allow for some less-than-maximally-parsimonious trees, it is a first step in that direction.
We provide a Python implementation with a flexible interface for the history sDAG as a container type for trees, endowed with abstract methods for convenient dynamic programs on the history sDAG structure, as well as all methods from this paper for manipulating history sDAGs constructed from maximally parsimonious trees.
This implementation shows the effectiveness of the approach, efficiently recovering many orders of magnitude more equally parsimonious trees than were used to ``seed'' the history sDAG when applied to a \sarscov\ data set.

\subsection{Intuitive Overview}

Here we provide an intuitive overview of the definitions and concepts used in this paper.
Formal definitions will be given in the sections following the overview.

This paper develops methods for understanding evolutionary relationships between samples from a population of closely related evolving entities, acknowledging uncertainty.
We will focus on samples consisting of nucleotide sequences, but keep our language general to emphasize that other data such as sample time and geographic location could also be used.

One way to formalize evolutionary relationships among samples, and inferred ancestral states, is to arrange them in a rooted phylogenetic tree with leaf and internal node labels.
Node labels in this tree can include data of the type associated with the given samples.
Specifically, leaf nodes are labeled by samples, and interior nodes are labeled by inferred ancestral states.
The set of samples which label leaves will be called the \emph{leaf labels}.
Interior node labels may be chosen from some larger label set which includes the leaf labels as a subset.
Instead of directly using this notion of a rooted, internally labeled tree, we will define a more convenient object called a \emph{history}, which holds the same data as such a tree.
We will make the definition formal below, but a history may be thought of as a rooted, internally labeled tree (this object has been called other names in the past, including an \emph{ancestral scenario}~\citep{Ishikawa2019}).
For example, a history might be used to represent a phylogenetic tree in which all (internal and tip) nodes are labeled with DNA sequences.

In a history, a node's \emph{clade} is the set of labels of its descendant leaf nodes (we emphasize that internal node labels are excluded from the clade definition).
A clade of a node's child is a \emph{child clade}.
The child clades of a node form a partition of the node's clade.
We therefore call this set of child clades a node's \emph{subpartition}.
Each edge in a history connects two nodes, each with a label and subpartition.
As a formality convenient for this paper, each history will contain a \emph{universal ancestor (UA) node} added as a parent of the root node.

Some histories explain the relationships between their leaf labels more plausibly than others.
One common measure of optimality for a history labeled by nucleotide sequences is its parsimony score, which is the total number of nucleotide base changes along all edges in the history.
A history is said to be \emph{maximally parsimonious} if no other history on the same leaf labels has a lower parsimony score.

In general, there are many possible maximally parsimonious histories with leaves labeled by the same set of nucleotide sequences.
We will use a structure called the \emph{history subpartition directed acyclic graph} (\emph{history sDAG}) to efficiently encode a large collection of histories (\autoref{fig:dag_construction}).
The ``history'' modifier emphasizes that this structure encodes a collection of possible rooted evolutionary histories, each of which contain not only a tree structure, but also ancestral state labels.

The history sDAG consists of a collection of nodes, each associated with a combination of label and subpartition, and one formal universal ancestor (UA) node, which is denoted $\rho$.
As we will see later, edges exiting $\rho $ keep track of the root nodes of the histories in the history sDAG\@.

A directed edge in a history sDAG represents an edge in a corresponding set of histories, from a parent node to a child node which have the same labels and subpartitions as the parent and child nodes of the edge in the DAG\@.
Thus, the history sDAG structure records combinations of labels and subpartitions, and adjacencies between these combinations, in the corresponding collection of histories.

By using a carefully chosen definition of history, introduced in the next section, the history sDAG can easily be constructed as the graph union of a set of histories.
These histories need not have identical leaf labels.
Specifically, we think of each history as its own history sDAG, with each node annotated by its label and subpartition.
The history sDAG constructed from the original set of histories is simply the union of nodes and edges in each history (\autoref{fig:dag_construction}).
The history sDAG then contains as subgraphs at least those histories used to construct it.

\begin{figure}[H]
    \centering
    \includegraphics{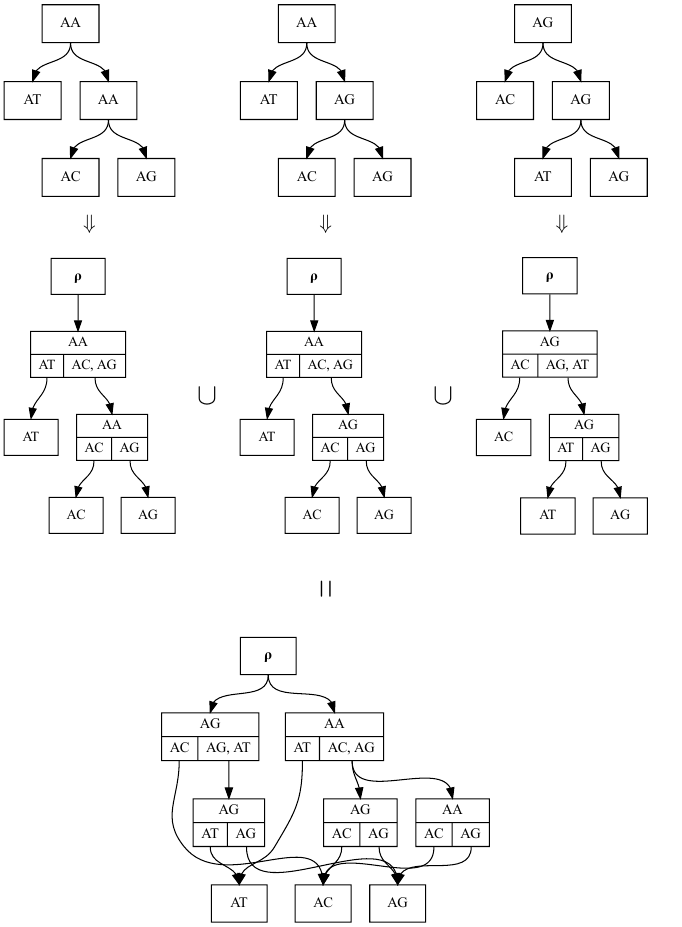}
    \caption{\
        A history sDAG constructed from three internally labeled trees on label set of sequences $\left\{AA, AC, AT, AG \right\} $.
        Each tree is converted to the equivalent history structure, and the union of these histories is the history sDAG\@.
        Each node in a history or the history sDAG consists of a label (in this case a sequence of two bases) shown in the top half of the node, and a subpartition, with each set in the subpartition separated by a vertical bar in the bottom half of the node.
    Leaf nodes have no children, so appear with only their label.
    Although in this example labels are length-two nucleotide sequences, the label set is arbitrary, and could include sequences, geographic location, or other information
    }\label{fig:dag_construction}
\end{figure}

Any subgraph of the DAG which is a tree, includes exactly one edge descending from the UA node, and exactly one edge descending from each child clade of each of its nodes, is a history (\autoref{fig:hdag_example_intro}).
Each of the histories contained in a history sDAG represents a combination of substructures from the histories used to construct the history sDAG\@.

In addition to thinking of the history sDAG as a way of recording structures observed in a collection of histories, we can also think of it as a way of generating histories.
In fact, the set of histories in the history sDAG is in general a superset of the set of histories used to construct the DAG (\autoref{fig:dagfindsmore}).
These new histories result from combining subhistories from histories used to construct the history sDAG\@.
This is similar to tree fusion, in which clades from different trees are combined to improve the parsimony score of the final tree~\citep{treefusion1999goloboff}.
The connection with tree fusion is explored further in the Discussion section.

\begin{figure}[h]
\centering
\includegraphics{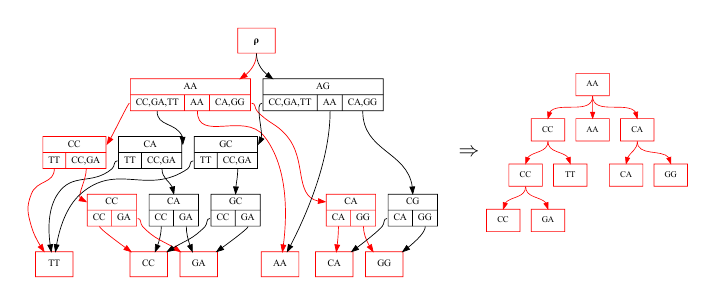}
\caption{\
    A history sDAG on label set $\left\{TT, CC, GA, AA, CA, GG, AG, CG\right\} $, with a history structure highlighted in red (left) and a labeled tree corresponding to that history (right)
}
\label{fig:hdag_example_intro}
\end{figure}

\begin{figure}[h]
\centering
\includegraphics{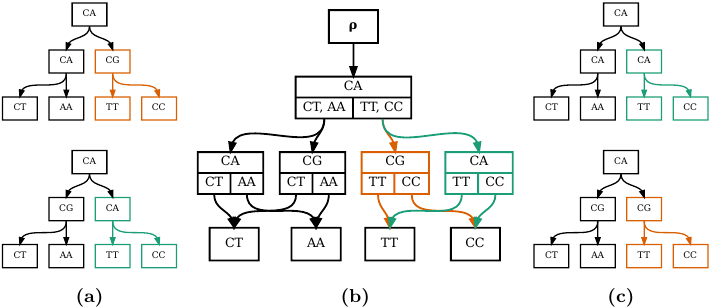}
\caption{\
The history sDAG can express more histories than were used to construct it.
The history sDAG in (b) is constructed from the two internally labeled trees in (a), and represents the four internally labeled trees in panels (a) and (c).
Notice that the trees in (c) are not among the trees used to construct the sDAG, but result from swapping the substructures highlighted in green and orange in (a)
}%
\label{fig:dagfindsmore}
\end{figure}

As described above, exploring phylogenetic uncertainty by examining many maximally parsimonious histories requires an efficient way to store and compute on those histories.
The history sDAG provides a compact structure for storing collections of histories, but in general contains histories beyond those used to build the DAG.
We therefore encounter a key question: are these additional histories also maximally parsimonious?

In the following two sections, we will show that maximum parsimony is in fact preserved by the history sDAG\@.
\autoref{thm:updown} shows that any history expressed by a history sDAG constructed from maximally parsimonious histories must itself be maximally parsimonious.
To achieve this we must first show that swapping certain substructures between histories preserves maximum parsimony.
Then we will show that the collection of histories in the history sDAG is closed under these subhistory swaps, and that any history in the history sDAG can be obtained by such a subhistory swap involving histories used to construct the DAG\@.
This means that the history sDAG is not only an effective way to store many maximum parsimony histories, but also may allow us to very quickly discover more such histories.

Preservation of maximum parsimony in the history sDAG has two important consequences:
\begin{itemize}
    \item A history sDAG constructed from maximally parsimonious histories will contain only maximally parsimonious histories.
        If a set $T $ of histories with the same parsimony score is used to construct a history sDAG, and if that history sDAG expresses a history with any other parsimony score, then $T $ must not have contained maximum parsimony histories.
    \item It is always possible to trim an arbitrary history sDAG to express all of, and only, its maximally parsimonious histories.
In particular, a new history sDAG constructed from the maximally parsimonious histories represented by the original history sDAG will contain only those histories used to construct it.
\end{itemize}

Throughout this paper we will refer to maximally parsimonious histories using the more general term \emph{minimum-weight} histories, since maximum parsimony is characterized by minimizing the sum of a weight over all edges in a history.
Indeed, this term is more general because we can use weight functions that are more complex than simply the sum of the number of mutations, or which consider label data other than nucleotide sequences.

We provide an implementation of the history sDAG and related algorithms in the open source Python package \historydag, installable with \texttt{pip} and available at \url{https://github.com/matsengrp/historydag}.
This package provides methods for constructing, trimming, collapsing, and extracting histories from the history sDAG as described in the following sections.
\historydag also implements methods which we will describe in future work, for efficiently calculating weights of histories represented in the history sDAG, and for expressing and sampling from a probability distribution on histories in the history sDAG\@.

For reference, we provide a summary of notation in~\autoref{tab:notation}.
\begin{table}[H]
    \begin{tabular}{l c l}
        \hline
        labels              & $\ell $         & a label, such as a nucleotide sequence\\
                            & $Y $            & a set of labels\\
                            & $X $            & the set of leaf labels, a subset of the label set $Y$\\
                            & $C $            & a clade, i.e.\ a subset of leaf labels\\
                            & $U $            & a set of disjoint clades\\
        \hline
        histories           & $t $            & a history\\
                            & $s $            & a subhistory of a history\\
                            & $ v $           & a node in a history or history sDAG\\
                            & $ s^v $         & a subhistory and its parent node\\
                            & $\rho $         & the UA node\\
                            & $ e $           & an edge in a history or history sDAG\\
                            & $T $            & a set of histories\\
                            & $f $            & an edge weight function on pairs of history\\
                            &                 & or history sDAG nodes\\
                            & $g_f $          & a weight function on histories, summing $f $\\
                            &                 & over all edges\\
                            & $L(s) $         & the set of leaf nodes of a subhistory $s $\\
                            & $\cu(v) $       & the clade union of a node $v $\\
        \hline
        history sDAGs        & $V $            & a set of history sDAG or history nodes\\
                            & $E $            & a set of history sDAG or history edges\\
                            & $D(T) $         & the set of histories expressed by a history sDAG\\
                            &                 & constructed from histories $T $\\
                            & $\tgts(v) $     & children of a node $v $ in a history or history sDAG\\
                            & $\shbelow(v)$   & subhistories below a node $v $ in a history sDAG\\
        \hline
    \end{tabular}
    \caption{Notation used in the text}\label{tab:notation}
\end{table}

\section{Histories and the History sDAG}
We will now provide a formal definition of histories and the history sDAG\@.

Let $Y $ refer to a set of labels, such as nucleotide sequences.
We can think of observed labels as a set $X\subset Y $, labeling history leaves.
We will not emphasize this set of leaf labels $X $, since a history sDAG may express collections of histories with varying leaf label sets.
In the case of parsimony however, we will be interested in collections of histories which share a leaf label set consisting of observed nucleotide sequences.

We are interested in representing collections of rooted, multifurcating, non-unifurcating trees with nodes (including internal nodes) labeled by elements of $Y $.
As mentioned in the Overview, we will make this easy by carefully defining histories.

Isomorphism classes of such internally labeled trees are in bijection with \textbf{histories}, as defined below.
This correspondence is shown formally in \hyperref[sec:correspondence]{Appendix~\ref*{app:proofs}}, but sufficient intuition may be found in \autoref{fig:dag_construction}.

Let $Y $ be a set of labels, and let $\mathcal{P}(\cdot)$ denote the power set.

\begin{defn}
Let $\Part(Y) $ be the set of all $U\subset \mathcal{P}(Y)\setminus \left\{\emptyset \right\} $ such that,
\begin{itemize}
    \item for $C_1, C_2\in U $, if $C_1 \neq C_2 $ then $C_1\cap C_2 = \emptyset $
    \item $|U| \neq 1 $.
\end{itemize}
That is, $\Part(Y) $ contains $\emptyset $ and all sets of two or more nonempty, disjoint subsets (\emph{clades}) of $Y $.
\end{defn}

Given a set of leaf labels $X\subset Y $, $\Part(X) $ would contain all of the possible subpartitions of leaf labels in an internally labeled tree with leaves labeled by $X $.
Notice that $\Part(X) \subset \Part(Y) $ for any such $X\subset Y $.
Since a history sDAG may contain histories with varying leaf label sets, elements of $\Part(Y) $ are used to construct general history sDAG nodes.

We will see that with the exception of a universal ancestor node, all nodes in the history sDAG structure consist of a label $\ell \in Y$ and a subpartition $U\in \Part(Y) $.
\begin{defn}
    A \textbf{node-clade pair} is a node $(\ell, U) $ and a choice of child clade $C\in U $.
\end{defn}

\begin{defn}\label{defn:historyDAG}
    A \textbf{history sDAG} with labels $Y $ is a directed graph $(V,E) $ consisting of
    \begin{itemize}
        \item A node set $V\subset \left ( Y\times \Part(Y)\right )\cup \left\{\rho \right\} $ such that $\rho \in V $ is the \textbf{universal ancestor (UA) node}.
            For a node $v = (\ell, U)\in V $, $v \neq \rho $, we say that $v $'s \textbf{label} is $\ell $, its \textbf{subpartition} is $U $, its \textbf{child clades} are elements of $U $, and its \textbf{clade union} $\cu(v) $ is $\left\{\ell \right\}$ if $U = \emptyset $, or $\bigcup\limits_{C\in U} C $ otherwise.
        \item A directed edge set $E\subset V\times V $ containing edges $e = (v_1, v_2) $ from a \textbf{parent node} $v_1 $ to a \textbf{target or child node} $v_2 $ such that
            \begin{enumerate}
                \item All nodes are reachable from the UA node $\rho $, which itself accepts no incoming edges.
                \item For any edge whose parent node is not $\rho $, the clade union of the target node must be in the subpartition of the parent node.

                    Formally, for any edge $e = \left((\ell_1, U_1), (\ell_2, U_2) \right)\in E $, if $C $ is the clade union of $(\ell_2, U_2) $, then $C\in U_1 $.

                        We say then that the edge $e $ \textbf{descends from the node-clade pair} $\left((\ell_1, U_1), C \right) $.
                \item For each node $v = (\ell, U) $, and for each choice of child clade $C\in U $, at least one edge descends from the node-clade pair $(v, C) $.
            \end{enumerate}
    \end{itemize}
\end{defn}

Notice that by requirements (1) and (3) in the definition of the history sDAG, all nodes in the history sDAG must have descendant edges, except for those of the form $(\ell, \emptyset) $.
We will refer to these as \textbf{leaf nodes}.

\begin{obs}
Since only nodes of the form $(\ell, \emptyset) $ may have no children, all leaf nodes in a history sDAG must be of this form, and therefore no two leaf nodes may be labeled by the same element of $Y $.
\end{obs}

\begin{obs}\label{obs:backwardsinclusion}
    For any history sDAG edge $(v_1 = (\ell_1, U_1), v_2) $, we know that $\cu(v_2) \subset \cu(v_1) $, since $\cu(v_2) \in U_1 $.
    More generally, consider a history sDAG $(V, E) $, in which a node $v' $ is reachable from another node $v $ via a sequence of edges in $E $.
    By transitivity of inclusion, $\cu(v') \subset \cu(v) $.
\end{obs}

\begin{defn}\label{defn:history}
    A \textbf{history} is a history sDAG in which the UA node $\rho $ has a unique child node, and each node-clade pair has exactly one descendant edge.

    The set of labels of the leaf nodes in a history $t $ will be denoted $L(t) $.
\end{defn}

Notice that not every element of $Y $ must appear as a node label in a history or history sDAG\@.
That is, $Y $ is an ambient label set, such as the set of all nucleotide sequences of a fixed length, from which history sDAG node labels can be chosen.

Also notice that there is no distinction between leaf node and internal node labels.
In practice, the set of leaf node labels will be associated with a set of observed evolving entities.
When a sampled entity is inferred to be an ancestor of other sampled entities, we can represent this in a history with an internal node carrying the label corresponding to the sampled ancestor.

Informally, a labeled tree can be converted to a history by annotating each node with its subpartition, and adding a UA node as a parent of the root node (\autoref{fig:dag_construction}).
The unique child of the UA node in a history will be called the \textbf{root node}, since it represents the root node of a corresponding internally labeled tree.

The natural substructure of a history is analogous to a subtree of a labeled tree, and will be very useful in later sections.
\begin{defn}
    Given a history sDAG $(V, E) $, a subgraph $s = (V_s, E_s) $ with $V_s \subset V $ and $E_s \subset E $ is a \textbf{subhistory} of $(V, E) $ if
    \begin{enumerate}
        \item $\rho \notin V_s $,
        \item there exists a \textbf{root node} $v_r\in V_s $ such that all other nodes in $V_s $ are reachable from $v_r $, and
        \item each node-clade pair in $s $ has exactly one descendant edge.
    \end{enumerate}

    The set of labels of leaf nodes in a subhistory $s $ is denoted $L(s) $.
\end{defn}
Later we will establish formally that a history is in fact a tree.
Given that fact, naming a subhistory is equivalent to removing an edge from a history, and discarding the component which contains the UA node.

In addition to the UA node, the definition of history contains redundant information in the sense that the subpartition of a node, formally a piece of data associated with each node, can be recovered as the set of sets of labels of leaf nodes reachable from that node's children.
Although this choice may seem an unnecessary complication, it is essential in distinguishing histories contained in a larger history sDAG\@.
This redundancy is shown in the following lemma, which is proven in \hyperref[proof:reachableleaves]{Appendix~\ref*{app:proofs}}:
\begin{restatable}[]{lemma}{reachableleaves}\label{lemma:reachableleaves}
    Let $(V,E) $ be a history sDAG or subhistory, and let $v\in V $.
    The set of labels of leaf nodes reachable from $v $ is $\cu(v) $.
\end{restatable}

\autoref{lemma:reachableleaves} implies that a history's set of leaf labels is determined by the subpartition of its root node.
This will be relevant later, when we describe what it means for a history to be found in a history sDAG.

We intend for a history to be tree-shaped, but this is not assumed by the definition given.
\autoref{lemma:reachableleaves} also allows us to prove this essential fact.

\begin{restatable}[]{lemma}{historyaltdef}
    \label{lemma:historyaltdef}
    A history sDAG $(V, E) $ is a history if and only if it is a tree, and contains exactly one edge descending from $\rho $.
\end{restatable}

The proof for this proposition is given in \hyperref[proof:historyaltdef]{Appendix~\ref*{app:proofs}}.

Notice that since elements of $\Part(Y) $ may not contain exactly one clade, and since nodes in a history have exactly one child node per child clade, no node (other than the UA node) in a history may have exactly one child.
This is required to ensure that the history sDAG may not contain cycles.
Although this is not stated in \autoref{defn:historyDAG}, it is an important property of the history sDAG as the name suggests, and is proven in \hyperref[proof:acyclic]{Appendix~\ref*{app:proofs}}:

\begin{restatable}[]{lemma}{acyclic}
\label{lemma:acyclic}
    A history sDAG $(V, E) $ is acyclic.
\end{restatable}

Sometimes data sets include a fixed root node label, such as a common ancestor sequence.
A search for minimum weight labeled histories explaining such a data set may yield labeled histories with a unifurcation at the root node.
We accommodate this by considering the fixed root sequence a leaf node label, and placing the corresponding leaf node as an additional child of the root node.

Since histories are tree-shaped history sDAGs, we can store collections of histories by taking their graph union.
However, we should first verify that a graph union of history sDAGs is itself a history sDAG\@.

\begin{lemma}\label{lemma:historyDAGunion}
    Let $(V, E) $ and $(V', E') $ be history sDAGs on labels $Y $.
    Then $(V\cup V', E\cup E') $ is also a history sDAG\@.
\end{lemma}
\begin{proof}
    All the nodes and edges required to satisfy \autoref{defn:historyDAG} are present in $(V\cup V', E\cup E') $, since they are present in each of the original history sDAGs.
    All nodes are reachable from the root node, through exactly the same sequence of edges by which they were reachable in at least one of the original histories.
\end{proof}

\begin{defn}\label{defn:historyDAGconstruct}
    For a set $T $ of histories with labels in $Y $, the \textbf{history sDAG constructed from} $T $ is the graph union of the histories in $T $:
    \begin{equation*}
        \left(\bigcup_{(V, E) \in T}V, \bigcup_{(V, E) \in T} E \right)
    \end{equation*}
\end{defn}

We should also formalize the way in which a history sDAG contains histories.
To do so, we will need to define a \emph{trim}, which is a history sDAG which appears as a substructure in a larger history sDAG.

\begin{defn}\label{defn:trim}
    Let $(V, E) $ be a history sDAG on labels $Y $.
    Then $(V', E') $ is a \textbf{trim} of $(V, E) $ if $V' \subset V $, $E'\subset E $, and $(V', E') $ is a history sDAG on labels $Y $.
    We say a history $t = (V'', E'') $ is \textbf{in} the history sDAG $(V, E) $ if $(V'', E'') $ is a trim of $(V, E) $.

    The collection of histories in the history sDAG constructed from a collection of histories $T $ will be denoted $D(T) $.
\end{defn}

We can now see why we must specify in \autoref{defn:history} that $\rho $ has exactly one child node in a history.
Edges descending from the UA node in a history sDAG keep track of which DAG nodes are allowed to be root nodes.
It may be possible to choose a tree-shaped trim of a history sDAG in which two nodes $v_1 $ and $v_2 $ are children of $\rho $, and $\cu(v_1) \cap \cu(v_2) = \emptyset $.
Such a structure should be considered a trim containing two histories, but is not itself a history.

Any history sDAG should be uniquely determined by the collection of histories it contains.
This intuition motivates the following two lemmas, which are proven in \hyperref[proof:nodesareroots]{Appendix~\ref*{app:proofs}}:

\begin{restatable}[]{lemma}{nodesareroots}\label{lemma:nodesareroots}
    Let $(V, E) $ be a history sDAG.
    For any $v\in V $, there exists a subhistory $s $ in $(V, E) $ whose root node is $v $.
\end{restatable}

\begin{restatable}[]{lemma}{dagisitshistories}
    \label{lemma:dagisitshistories}
    Let $(V, E) $ be a history sDAG, and let $T $ be the collection of histories in $(V, E) $.
    Then $(V, E) $ is the history sDAG constructed from $T $.
\end{restatable}

Finally, we will need to define the largest possible history sDAG constructed using a given set of labels.
\begin{defn}\label{defn:complete}
    The \textbf{complete history sDAG} on labels $Y $ is the history sDAG which contains all possible edges on all nodes allowed by the choice of $Y $.
\end{defn}

Equivalently, the complete history sDAG could be constructed as the graph union of all possible histories with labels in $Y $.

\subsection{History Weights}
In this section we will define a general scheme for assigning weights to histories, and describe the relationship between these weights and the structure of the history sDAG.

As shown in~\autoref{fig:dagfindsmore}, the history sDAG in general contains more histories than were used to construct it.
These extra histories arise because the history sDAG allows subhistories between the histories it contains, whenever the subhistories' parent nodes share the same child clades and node label.
We refer to this occurrence as \textbf{subhistory swapping}.
\hyperref[proof:updownthm]{Appendix~\ref*{defn:subhistoryswap}} describes these subhistory swaps precisely, shows that all new histories in a history sDAG can be described as in terms of sequences of these subhistory swap operations, and provides the proof for~\autoref{thm:updown}, which involves an argument that these subhistory swaps preserve history weights.

This section will leave the details of subhistory swaps, and the proof of~\autoref{thm:updown}, to the Appendix, and only build the background necessary to state and understand~\autoref{thm:updown}.

We begin by defining another useful type of history substructure.

\begin{defn}
    Let $(V, E) $ be a history sDAG and let $s = (V_s, E_s) $ be a subhistory $(V, E) $.
    Also, let $v_r $ be the root node of the subhistory $s $, and let $v\in V $ be a parent node of $v_r $, so that $(v, v_r) \in E $.
    Then the \textbf{augmented subhistory} $s^{v} $ is the subgraph $\left(V_s\cup \left\{v\right\}, E_s \cup \left\{(v, v_r) \right\} \right) $ of $(V, E) $ consisting of the subhistory $s $ plus the parent node $v $ and the edge connecting $v $ to $v_r $.
\end{defn}
\begin{defn}
    Let $(V, E) $ be a history sDAG, and $v = (\ell, U) \in V $ a node.
    We make the following definitions.
    \begin{itemize}
    \item $\tgts(v) := \left\{v_c \mid (v, v_c) \in E \right\} $ will denote the set of children of $v $
    \item $\tgts(v, C) := \left\{v_c \mid (v, v_c)\in E,\ \cu(v_c) = C \right\} $ will denote the set of children of the node-clade pair $(v, C) $ for each clade $C\in U$
    \item $\shbelow(v) $ will denote the set of subhistories in $(V, E) $ rooted at $v $.
    \end{itemize}
\end{defn}

Although we are interested primarily in computing parsimony on histories labeled with nucleotide sequences, we will do so within a much more general framework of \textbf{history weights}.
\begin{defn}
    Let $(V, E) $ be the complete history sDAG on labels $Y $, and let $f:E\to W $ be an \textbf{edge weight function} to a weight set $W $ endowed with addition and containing an additive identity $0\in W $.
    The \textbf{weight} of any subgraph $(V', E') $ of $(V, E) $ is then given by the \textbf{weight function} $g_f$
\begin{equation*}
    g_f\left((V', E') \right) = \sum\limits_{e \in E'} f(e).
\end{equation*}
    In particular, since any history $t $ in $(V, E) $ is a subgraph of $(V, E) $, the weight of $t $ is given by $g_f(t) $.
\end{defn}

In the case of parsimony, the label set $Y $ will contain sequences, the function $f $ is Hamming distance, and $g_f $ will compute the parsimony score of a history.
A history's parsimony score is decomposable as a sum of an edge weight function over edges only when complete, unambiguous nucleotide sequences are accessible to that weight function as node label data.
If nucleotide sequences of internal nodes are not contained in node label data, the contribution of an edge to a history's parsimony score may be dependent on the structure of the rest of the history, making the decomposition impossible.
In particular, the edge weight function $f $ is required to be a function on all possible history sDAG edges, which correctly reports the contribution of an edge to the weight of any history which contains it.

Although our focus here is parsimony, notice that this framework allows much more general notions of history weight, including situations where the function $f $ is sensitive to edge direction or subpartitions, or takes values in a non-numeric set, such as a set of sequences.
These generalizations will be important for future applications.
For example, we could compute a branching process likelihood like that used by the \gctree project, whose value can be decomposed over tree edges, and which can be summarized by a pair of integers~\citep{dewitt2018using}.

To compare weights of histories, the weight set $W $ must admit a total ordering.
This ordering will be required to respect addition on $W $, in a slightly weaker sense than is generally meant:

\begin{defn}
    A weight set $W $, endowed with addition, is \textbf{clade-ordered} with respect to some edge weight function $f $ and history sDAG $(V, E) $ on labels $Y $ if
    \begin{itemize}
        \item The ordering on $W $ respects addition and is a total ordering on all of the following subsets of $W $:
        \begin{itemize}
            \item Sets of weights of subhistories below any node:

                $\left\{g_f(s) \mid s\in \shbelow(v) \right\}$, for any $v \in V \setminus \left\{\rho \right\} $,
            \item Sets of weights of augmented subhistories below any node-clade pair:

                $\left\{g_f(s^v) \mid s\in \shbelow(v_c),\ v_c\in \tgts(v, C)\right\}$, for any $v = (\ell, U) \in V\setminus \left\{\rho \right\} $, and any $C\in U $.
        \end{itemize}

        \item The ordering on $W $ is a total ordering on the set of weights of histories:
        \[\left\{g_f(t) \mid t\text{ is a history in } (V, E)  \right\} \subset W\]
    \end{itemize}

    We say that the ordering on $W $ respects addition on a set $W'\subset W $ if for all $a, b\in W' $ and for all $c\in W $, $a < b $ if and only if $a + c < b + c $.
\end{defn}

The following observation makes this definition easier to use.

\begin{obs}
   Let $W $ be a weight set which is clade-ordered with respect to a history sDAG $(V,E) $ and edge weight function $f $.
   If $(V', E') $ is a trim of $(V, E) $, and $f': E' \to W$ is equal to $f $ restricted to $E' $, then $W $ is also clade-ordered with respect to $f' $ and $(V', E') $.
\end{obs}

For example, it may often be most convenient to argue that a weight set is clade-ordered with respect to the complete history sDAG on the label set $Y $, and a weight function defined on all possible edges in that history sDAG.

However, since this is a strictly stronger condition on $W$, which is why the definition of clade-ordering is with respect to a particular history sDAG.

Through the rest of this section, the label set $Y $ will be fixed, and it will be assumed that $f $ is an edge weight function mapping into $W $, a weight set which is clade-ordered with respect to $f $.

Finally, we can describe exactly the sense in which the history sDAG preserves history weights, a property depicted in \autoref{fig:updownandcor}.

\begin{figure}[H]
    \centering
    \includegraphics{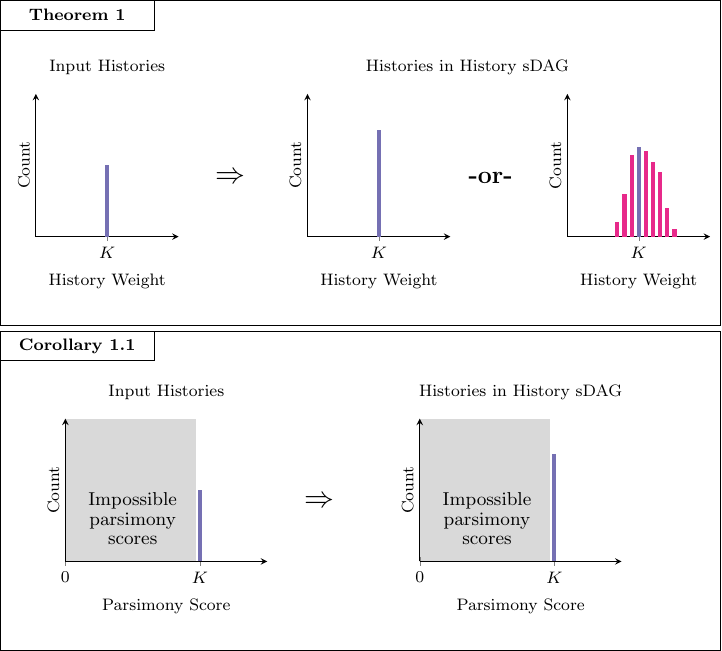}
    \caption{\
        \autoref{thm:updown} states that if a history sDAG is built from a collection of histories which all have weight $K $, then either the resulting sDAG must contain only histories of weight $K $, or there must be histories with weights greater \textbf{and} less than $K $.
        In either case the resulting history sDAG may contain more histories than were used to build it.
        \autoref{cor:trimDAG} observes that since no parsimony score less than the maximum parsimony score can be achieved by a history on a given leaf set, a history sDAG built from maximally parsimonious histories must contain only maximally parsimonious histories.
}%
\label{fig:updownandcor}
\end{figure}
\begin{restatable}[]{thm}{updownthm}\label{thm:updown}
    Let $T $ be a collection of histories, so that $g_f(t) = K $ for all $t\in T $.
    Then there exists a history $t\in D(T) $ with $g_f(t) < K $ if and only if there exists a history $t'\in D(T) $ with $g_f(t') > K $.
\end{restatable}

\autoref{thm:updown} is the motivation for and main result of this section, guaranteeing that a history sDAG constructed from minimum weight histories will only express minimum weight histories, and is proven in \hyperref[proof:updownthm]{Appendix~\ref*{app:proofs}}.

However, since it may be impractical to verify that a collection of histories are minimum weight relative to all other possible histories on a chosen label set, \autoref{thm:updown} will often be more useful when applied in the form of the following corollary, that any history sDAG may be trimmed to express exactly its minimum weight histories, relative only to the other histories in that history sDAG\@.

\begin{cor}\label{cor:trimDAG}
    Let $(V, E) $ be a history sDAG, and let $f $ be an edge weight function as defined previously.
    Then there exists a history sDAG $(V', E') $ which is a trim of $(V, E) $ such that the histories in $(V', E') $ are exactly the minimum weight histories in $(V, E) $ with respect to $f $.
\end{cor}
\begin{proof}
    Let $T $ be the collection of histories expressed by $(V, E) $, so that $D(T) = T $.
    Let $K $ be the minimum weight achieved by $g_f $ on $T $, and let $T' \subset T $ be the set of minimum weight histories:
    \begin{equation*}
        T' = \left\{t\in T \st g_f(t) = K \right\}
    \end{equation*}
    We know that $T'\subseteq D(T') $, so we need only show that $T'\supseteq D(T') $.
    Since $T' \subseteq T $, we know that $D(T') \subseteq D(T) = T$, and that since $T $ is the collection of histories in $(V, E) $, there exists no history $t\in D(T') $ with $g_f(t) < K $.
    Therefore, by \autoref{thm:updown}, there exists no $t\in D(T') $ with $g_f(t) > K $.
    Since $T' $ contains all the histories in $T $ with weight $K $, we therefore know that $T' = D(T') $.
    Let $(V', E') $ be the history sDAG constructed from $T' $.
    Since $T' = D(T') $, the history sDAG $(V', E') $ contains exactly the histories in $T' $.
    Also, because $(V', E') $ is a graph union of histories in $(V, E) $, we know that $V'\subset V $ and $E'\subset E $.
    Therefore $(V', E') $ is the trim of $(V, E) $ that we seek.
\end{proof}

We shall take a small excursion now, in which we return to the setting of maximum parsimony which motivates these methods.
It makes little sense to minimize parsimony on the set of all histories with labels in an ambient sequence set $Y $.
Rather, one attempts to minimize parsimony subject to the constraint that history leaves are labeled by some fixed set of observed nucleotide sequences.

\begin{defn}
    Let $T $ be a set of histories with labels in $Y $.
    We say that histories in $T $ have a \textbf{fixed set} $X\subset Y $ of leaf labels if $L(t) = X $ for all $t\in T $.

    Given an edge-weight function $f $ and a set $X\subset Y $, we say that a history $t $ with $L(t) = X $ is \textbf{minimum weight relative to all histories on the fixed set of leaf labels} $X $ if $g_f(t) \leq g_f(t') $ for all histories $t' $ with $L(t') = X$.
\end{defn}

In the general language of this section, a history $t $ with nucleotide sequence labels is maximally parsimonious if it is minimum weight relative to all histories on the fixed leaf label set $L(t) $, with Hamming distance as the edge-weight function.

The following observation guarantees that \autoref{thm:updown} and \autoref{cor:trimDAG} are useful in this setting.

\begin{obs}\label{obs:fixedleaves}
    Let $T $ be a set of histories with a fixed set of leaf labels $X\subset Y $.
    Then for any $t\in D(T) $,  $L(t) = X $.
\end{obs}

The truth of this observation can be argued precisely using the lemmas in \hyperref[proof:updownthm]{Appendix~\ref*{app:proofs}} supporting the proof of \autoref{thm:updown}, but is apparent from \autoref{defn:historyDAG} and \autoref{fig:dag_construction}.

This means that given a set $T $ of maximally parsimonious histories on a fixed set of leaf labels $X $, $D(T) $ must only contain histories with leaves labeled by $X $.
By \autoref{thm:updown} then, $D(T) $ must only contain histories which are maximally parsimonious on leaf labels $X $.

If $T $ contains histories on a fixed label set $X $ which are not necessarily maximally parsimonious, \autoref{obs:fixedleaves} ensures that trimming the history sDAG constructed from $T $ as in \autoref{cor:trimDAG} will result in a new history sDAG which expresses histories with the same fixed set of leaf labels $X $.

\subsection{Trimming the history sDAG}
Here we describe a straightforward method for trimming a history sDAG to represent only its minimum-weight histories.
\autoref{cor:trimDAG} guarantees that merging only the minimum-weight histories in a history sDAG will result in a new history sDAG containing only those histories, but provides no efficient method for producing this trimmed history sDAG.
The method described here involves removing all edges which point to suboptimal subhistories, and can be realized in two traversals of the history sDAG.

\begin{defn}
Let $(V, E) $ be a history sDAG on labels $Y $, and let $f $ be an edge-weight function $f: E\to W $ for $W $ a weight set which is clade-ordered with respect to $f $ and $(V, E) $.

The minimum weight of an augmented subhistory beneath a node $v=(\ell, U) \in V $ and a clade $C\in U $ is given by $M_f(v, C) $, defined as
\begin{equation*}
    M_f(v, C) = \min \left\{g_f(s^v) \st v_c \in \tgts(v, C), s \in \shbelow(v_c) \right\}.
\end{equation*}

Also let $M_f(v) $ report the minimum weight of any subhistory rooted at the node $v=(\ell, U) $, and for any leaf node $v'\in V $, let $M_f(v') $ be the additive identity of $W $.
\end{defn}

Notice that because $W $ is clade-ordered, $M_f(v) $ can be computed as
\begin{equation}\label{eqn:minsubhistorybelownode}
    M_f(v) = \sum_{c\in U} M_f(v, C).
\end{equation}

That is, the minimum weight of a subhistory beneath a node is given by the sum over clades of the minimum weight achieved by an augmented subhistory below each clade.

Notice that the clade-ordering on $W $ also allows us to compute $M_f(v, C) $ more easily, as
\begin{equation*}
    M_f(v, C) = \min \left\{M_f(v_c) + f(v, v_c) \st v_c \in \tgts(v, C)\right\}.
\end{equation*}

With \autoref{eqn:minsubhistorybelownode}, this defines an efficient dynamic program for calculating the minimum weight of all histories in a history sDAG with respect to $f $, with

\begin{equation*}
    M_f(\rho) = \min\left\{M_f(v_c) + f(\rho, v_c) \st v_c \in \tgts(\rho) \right\}.
\end{equation*}

$M_f $ will be used to define the trimmed history sDAG:

\begin{defn}
    Let $(V, E)  $ be a history sDAG and $f:V\to W $ be an edge weight function, with $W $ clade-ordered.
    The \textbf{minimum weight trim} of $(V, E) $ with respect to $f $ is defined to be $(\underline{V}, \underline{E}) $, where
    \begin{align*}
        \underline E' &= \left\{(v, v_c)\in E \st M_f(v_c) + f(v, v_c) = M_f(v, \cu(v_c)) \right\},\\
        \underline V &= \left\{v\in V\st v \text{ reachable from } \rho \text{ via a path in } \underline E' \right\}, \text{ and}\\
        \underline E &= \left\{(v, v_c)\in \underline E' \st v, v_c \in \underline V \right\}.
    \end{align*}
\end{defn}
Notice that $\underline E'$ consists of edges from $E $ which point to optimal subhistories, $V' $ contains nodes reachable from $\rho $ via those edges, and $\underline E $ removes edges from $\underline E' $ which connect any nodes not in $\underline V $.

The following lemma verifies that this structure is what its name suggests.
\begin{restatable}[]{lemma}{trimistrim}\label{lemma:trimistrim}
    Let $(V, E) $ be a history sDAG, and $f:E\to W $ be an edge-weight function, with $W $ a weight set which is clade-ordered with respect to $f $ and $(V, E) $.
    Let $(V', E') $ be the history sDAG constructed from minimum-weight histories in $(V, E) $, with respect to $f $, and let $(\underline V, \underline E) $ be the minimum weight trim of $(V, E) $ with respect to $f $.
    Then $(V', E') = (\underline V, \underline E) $.
\end{restatable}
The proof for this lemma is given in \hyperref[proof:trimistrim]{Appendix~\ref*{app:proofs}}.

\subsection{Collapsing histories}
The space of possible minimum weight histories on a fixed leaf label set is in general very large.
However, some diversity in this set is a result of unnecessary history edges between nodes with the same label.
Unless these edges target a leaf node, they are unnecessary, and their existence cannot be supported by the observed data represented in leaf labels.

Just as polytomies can be resolved as many possible bifurcating structures, collapsing history edges which connect nodes with identical labels reduces the number of possible histories on a fixed set of leaves, without restricting the number of informative evolutionary scenarios that can be expressed by those histories (\autoref{fig:collapse_motivation}).

\begin{figure}[h]
    \centering
    \includegraphics{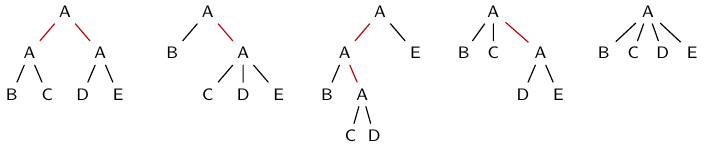}
    \caption{\
    By collapsing red edges between nodes with identical labels, all five internally labeled tree structures shown here are equivalent
}%
\label{fig:collapse_motivation}
\end{figure}

Motivated by this observation, we will enforce in practice that adjacent nodes in a history not have the same label, unless one of them is a leaf node.
This choice is possible because we allow multifurcations in histories, which leads to the definition of ``collapsing'' below.
On the other hand, sampled ancestors in a history can be witnessed as an internal node with the observed label $\ell \in Y $, adjacent to the leaf node labeled $\ell $.
Since the edge between these two nodes targets a leaf, such a structure is allowed in a history.

A history containing internal edges whose parent and child nodes carry the same label may be modified to remove such edges.
Doing so will add multifurcations to the history, as shown in~\autoref{fig:collapse_in_history}.
The following definition allows us to mark edges as collapsible arbitrarily, not just when their parent and child node labels match.
This generality is useful in precisely stating~\autoref{lemma:collapseone}.

\begin{defn}
    Let $(V, E) $ be a history or history sDAG.

    Given a binary-valued function $b: E\to \left\{0,1 \right\} $, an edge $e = \left((\ell, U), (\ell', U') \right) \in E $ is \textbf{$b $-collapsible} if $b(e) = 1 $ and $U' \neq \emptyset $ (so the target node is not a leaf node).
    An edge is \textbf{$b $-collapsed} if it is not $b $-collapsible.
    $(V, E) $ is $b $-collapsed if each edge in $E $ is $b $-collapsed.

    For the purpose of this paper we are interested in collapsing edges whose parent and child nodes have the same label.
    In this situation $b $ should return $1 $ on edges whose parent and child nodes have the same label, and we will use the terms \textbf{label-collapsible} and \textbf{label-collapsed} instead of $b $-collapsible and $b $-collapsed.
\end{defn}

A history which is not label-collapsed can be converted to a label-collapsed history by merging adjacent nodes with the same label, but this process requires also modifying subpartitions (\autoref{fig:collapse_in_history}).

\begin{figure}[H]
    \centering
    \includegraphics{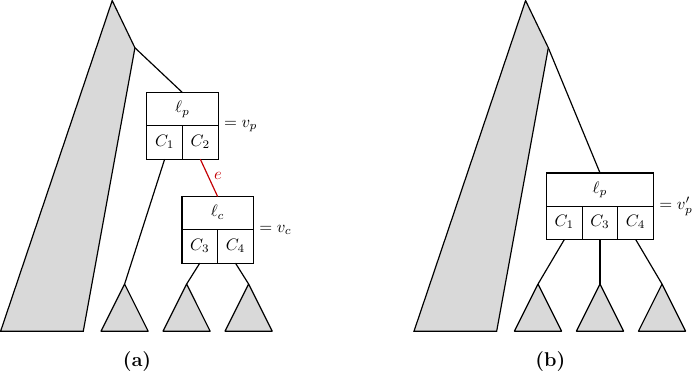}
\caption{\
    (a) shows part of a history, with an edge $e $ to be collapsed.
    Collapsing $e $ requires replacing the child clade $C_2 $ of $v_p $ with the child clades of $v_c $, to create the new combined node $v_p' $ (b)
}%
\label{fig:collapse_in_history}
\end{figure}
To formalize this, we first explain what it means to collapse an edge in a history.

\begin{defn}
    Let $t = (V_t, E_t) $ be a history with labels $Y $.
    Let $(V, E) $ be the complete history sDAG on labels $Y $.
    Also let $e = \left((\ell_p, U_p), (\ell_c, U_c) \right)\in E_t $ be an edge in $t $, so that $(\ell_c, U_c) $ is not a leaf node.
    Let $C = \cu(\ell_c, U_c) $ be the clade in $U_p $ from which the edge $e $ descends.

    The history $t_e = (V_e, E_e) $, formed by collapsing $e $ in $t $, is defined as follows:

    Define $q:V_t \to V $ via
    \begin{equation*}
        q(v) = \begin{cases}
            \left(\ell_p, U_p \cup U_c \setminus \left\{C \right\} \right) & v = (\ell_p, U_p) \text{ or } v=(\ell_c, U_c)\\
            v & \text{otherwise}
        \end{cases}
    \end{equation*}

    Then $V_e = q(V_t) $, and $E_e = \left\{(q(v), q(v')) \st (v, v') \in E_t\setminus \left\{e \right\} \right\} $.
\end{defn}
Notice that after collapsing an edge, the resulting structure remains a valid history, because for any clade $C\in U_c $, and for any node $v_c $ which is a child of the node-clade pair $\left((\ell_c, U_c), C \right) $, the node $v_c $ becomes a child of the node-clade pair $\left(q(U_c, \ell_c), C \right) $.
Also notice that $q(\ell_c, U_c) = q(\ell_p, U_p) $ inherits the unique parent of $(\ell_p, U_p) $ in $t $.

The new history has one edge fewer than the original.

We can convert a history $t $ to a \emph{label-collapsed history} by iteratively collapsing each edge in $t $ whose parent and child nodes have the same label.

\begin{lemma}
    A history $t_0 = (V_0, E_0)$ determines a unique label-collapsed history $t_c $, which is the result of a finite sequence of edge collapses.

    That is, there exists a finite sequence $t_0, t_1 = (V_1, E_1), \ldots, t_n = (V_n, E_n) $ for which
    \begin{itemize}
        \item $t_i $ is the result of collapsing some edge $e_i = \left((\ell_i, U_i), (\ell_i', U_i') \right)$ in $t_{i-1} $ for which $\ell_i = \ell_i' $ and $U_i'\neq \emptyset $, and
        \item $t_n $ is label-collapsed
    \end{itemize}
    Furthermore, for any such sequence of histories, $t_n = t_c $.
\end{lemma}
\begin{proof}
    Using the correspondence between histories and rooted, internally labeled, multifurcating trees established in \hyperref[lemma:correspondence_bijection]{Appendix~\ref*{app:proofs}}, we can use the fact that collapsing edges between internal nodes with the same label is a well-defined map on such trees.
    Since the order of edge collapse has no effect on the final tree, neither does the order of edge collapse on the final history in the sequence named above.
\end{proof}

Label-collapsing histories individually is straightforward, but collapsing a large collection of histories could be done more efficiently by label-collapsing their history sDAG\@.

Label-collapsing histories from within a history sDAG is not as straightforward, because some edges descending from a node-clade pair may need to be collapsed, while others may not.
This means that an algorithm to collapse the history sDAG must occasionally add new nodes to the DAG (\autoref{fig:collapse_in_DAG}).

\begin{figure}[H]
    \centering
    \includegraphics{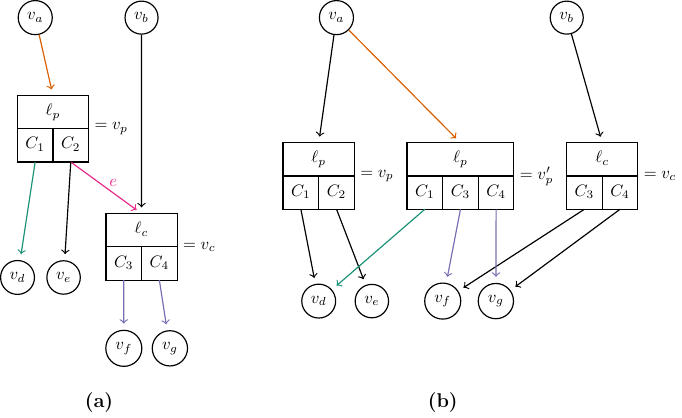}
\caption{\
    Analogous to \autoref{fig:collapse_in_history}, but within a history sDAG, (a) shows part of a history sDAG, with an edge $e $ to be collapsed.
    Collapsing $e $ requires adding the node $v_p' $ (b).
    In this example, both $v_p $ and $v_c $ remain in the history sDAG, because even without $e $ they each have a parent edge, as well as one child edge descending from each child clade.
    Edges in (a) are colored to match with the corresponding new edges in (b), and with the annotations in \autoref{eqn:collapseone}
}%
\label{fig:collapse_in_DAG}
\end{figure}

In order to describe the behavior of collapsing in the history sDAG, we require the following definition.

\begin{defn}\label{defn:collapsibletreecover}
    Given a history sDAG $(V, E)$, we say that a collection of histories $T$ is an edge cover of $(V, E)$ if for every edge $e\in E$, there exists a history $t\in T$ such that $e$ is contained in $t$.

    Further, a collection of histories $T $ is a \textbf{$b $-collapsible edge cover} of $(V, E)$ if
    for every $b $-collapsible edge $e\in E$ and every subhistory $s$ containing $e$, there is a history $t\in T$ that contains $s$.
\end{defn}

The following lemma describes what it means to collapse a single edge in a history sDAG.

\begin{restatable}[]{lemma}{collapseone}\label{lemma:collapseone}
Let $(V, E) $ be a history sDAG with label set $Y $.

Also let $(v_p = (\ell_p, U_p), v_c = (\ell_c, U_c))\in E $ be an internal edge.
That is, $U_c\neq \emptyset $ (so $v_c $ isn't a leaf node).

Define a binary function $b: E\to \left\{0, 1 \right\} $ which is constant at $0 $, except that $b(v_p, v_c)=1 $, and let $T $ be any $b $-collapsible edge cover of $(V, E) $.

Let $v_p' = (\ell_p, U_p\cup U_c \setminus \cu(v_c)) $ be the ``new parent node'', and define:
\begin{align}\label{eqn:collapseone}
    E^+ =\ & E & \\
         & \cup \left\{(v, v_p') \st (v, v_p) \in E \right\} & \text{\color{brewer-orange}``grandparents''}\nonumber\\
         & \cup \left\{(v_p', v) \st (v_p, v) \in E, \cu(v) \neq \cu(v_c) \right\} &\text{\color{brewer-green}``children''}\nonumber\\
         & \cup \left\{(v_p', v) \st (v_c, v) \in E \right\} &\text{\color{brewer-purple}``grandchildren''}\nonumber\\
         & \setminus \left\{(v_p, v_c) \right\} &\text{\color{brewer-pink}``collapsed edge''}\nonumber
\end{align}

Let $R = \emptyset$ if there exists an edge $(v_p, v)\in E^+ $ with $\cu(v) = \cu(v_c) $.
Otherwise, let $R $ be the set of parent edges of $v_p $, of the form $(v', v_p) \in E^+ $.

Then, let $E^- = E^+\setminus R $.
Finally, define
\begin{equation*}
    E' = \left\{ (v_1, v_2) \st (v_1, v_2) \in E^-,\ v_1 \text{ reachable from $\rho $ via edges in } E^- \right\}
\end{equation*}
and
\begin{equation*}
    V' = \left\{v_1, v_2 \st (v_1, v_2) \in E' \right\}.
\end{equation*}

    \textbf{Claim:} $(V', E') $ is the history sDAG constructed from $T' $, the set of histories which result by collapsing the edge $(v_p, v_c) $ in each history in $T $ in which it appears.
\end{restatable}

Notice that if $e $ is the only edge descending from the node-clade pair $(v_p, \cu(v_c)) $, then collapsing $e $ requires removing the node $v_p $, and all edges involving it, from the history sDAG\@.
Also, the definition of $E^- $ will not leave any parent nodes of $v_p $ with too few descendant edges, because we added edges from all parent nodes of $v_p $ to $v_p' $.

The last step in the construction of $E' $ ensures that any nodes left without parents in the collapsing process will not appear in the label-collapsed history sDAG\@.

The proof for \autoref{lemma:collapseone} is given in \hyperref[proof:collapseone]{Appendix~\ref*{app:proofs}}.

Finally we arrive at the main result of this section, which provides a guarantee that all histories in a history sDAG can be collapsed by a finite sequence of edge collapses.
Although this lemma is stated for label-collapsing, the result can immediately be generalized to $b $-collapsing, with respect to an arbitrary binary function $b $.

\begin{restatable}[]{lemma}{collapsingterminates}
    \label{prop:collapsingterminates}%
    Let $(V_0, E_0) $ be a history sDAG, and define a sequence $(V_i, E_i)_{i\in \N} $ of history sDAGs, so that $(V_k, E_k) $ is generated by collapsing an edge
    \[e_{k-1} = \left((\ell_{k-1}, U_{k-1}), (\ell_{k-1}', U_{k-1}')\right)\]
    in $(V_{k-1}, E_{k-1}) $ with $\ell_{k-1} = \ell_{k-1}' $ and $U_{k-1}' \neq \emptyset $ if such an edge exists.
    If no such edge exists, then $(V_k, E_k) = (V_{k-1}, E_{k-1}) $.

    Then there exists $N\in \N$ such that $(V_N, E_N) $ is label-collapsed.
    Also, if $T_0 $ is a label-collapsible edge cover of $(V_0, E_0) $, and $T_0' $ is the set of histories resulting from label-collapsing each history in $T_0 $, then each history in $T_0' $ is in $(V_N, E_N) $.
\end{restatable}
Although this lemma is written for label-collapsing, it extends to collapsing with respect to an arbitrary binary function $b $, defined on all possible edges in the complete history sDAG with the same leaf nodes and with labels chosen from the same ambient label set as $(V_0, E_0) $.

Note that the collapsing algorithm presented below produces the collapsed history sDAG $(V_N, E_N) $.

The proof for this proposition is given in \hyperref[proof:collapsingterminates]{Appendix~\ref*{app:proofs}}.

\autoref{prop:collapsingterminates} suggests an algorithm for collapsing a history sDAG, whose implementation is given below.

\begin{taocpalg}{A}{Collapsing a history sDAG}{%
        Modifies a history sDAG so that no edges connect two non-leaf nodes with the same label, and the histories represented in the resulting history sDAG are the same as the set of histories represented by the original history sDAG, with each label-collapsed.
    }

    \begin{enumerate}
        \item \textbf{Build queue.} $\mathcal{Q} := (v_i, v_i')_{i=1}^{|E|}$ is a queue of edges in $(v_i, v_i') \in E $ so that if $(v_i, v_i') $ and $(v_j, v_j') $ are such that $v_i' = v_j $, then $j > i $.
    That is, edges at the beginning of the queue are closer to the UA node of the history sDAG

        \item \textbf{Collapse loop head.} If $\mathcal{Q} $ is empty, \textbf{END}.
Otherwise, remove the first element $\left(v_p = (\ell_p, U_p), v_c = (\ell_c, U_c) \right) $ from $\mathcal{Q} $.

        \begin{enumerate}
        \item \textbf{Check collapsed.}
        If $\ell_p = \ell_c $ and $v_c $ is not a leaf node, and $(v_p, v_c)\in \mathcal{Q}$ , go to \textbf{new parent}.
        Otherwise, return to \textbf{collapse loop head}.

    \item \textbf{New parent.} Set $v_p' := \left(\ell_p, U_p \cup U_2 \setminus \left\{\bigcup\limits_{C\in U_c} C \right\} \right) $. Add $v_p'$ to $V$.

    \item \textbf{Add grandparents to newparent.} For any $(v, v_p)\in E $, add $(v, v_p') $ to $E $ and to beginning of $\mathcal{Q} $.

    \item \textbf{Add children to newparent.} For any $(v_p, v)\in E_t$, if clade union of $v $ is not the same as the clade union of $v_c $, add $(v_p', v) $ to $E $ and to beginning of $\mathcal{Q} $.

    \item \textbf{Add grandchildren to new parent.} For any $(v_c, v) \in E $, add $(v_p', v) $ to E and to the beginning of $\mathcal{Q} $.

    \item \textbf{Remove collapsed edge.} Remove $(v_p, v_c) $ from $E $.

    \item \textbf{Remove lonely parent.}
    If no edge $(v_p, v)\in E $ exists with clade unions of $v $ and $v_c $ equal, then do routine removenode $v_p $ from $(V, E)$.

    \item \textbf{Remove orphaned child.}
    If no edge $(v, v_c)\in E $ exists, do routine \textbf{removenode} $v_c $ from $(V, E) $.
    Return to \textbf{Collapse loop head}.

        \end{enumerate}
    \end{enumerate}

    The routine \textbf{removenode} $v$ from $(V, E) $ is the following:
    \begin{enumerate}
        \item \textbf{Remove node.} Remove $v $ from $V $.
        \item \textbf{Remove children loop head.} For each child node $v_c $ of $v $:
            \begin{enumerate}
                \item \textbf{Remove edge.} Remove the edge $(v, v_c) $ from $E $.
                \item \textbf{Clean child node.} If no edge $(v_p, v) $ exists in $E $, then do routine \textbf{removenode} $v_c$ from $(V, E) $.
            \end{enumerate}
        \item \textbf{Remove parent loop head.} For each parent node $v_p $ of $v $:
            \begin{enumerate}
                \item \textbf{Remove edge.} Remove the edge $(v_p, v) $ from $E $.
            \end{enumerate}
    \end{enumerate}
\end{taocpalg}

Notice that each iteration of the collapse loop corresponds with an element in the sequence of history sDAGs named in \autoref{prop:collapsingterminates}.
Since the order of edges in the sequence $(e_k) $ in \autoref{prop:collapsingterminates} has no effect on the resulting history sDAG, the order of edges in the queue should have no effect on the history sDAG produced by this algorithm.

\subsection{History sDAG Completion}
We now introduce ``completion,'' which essentially means that we add every edge that respects clade union sets.
More precisely, \autoref{defn:historyDAG} specifies that each edge of a history sDAG must target a node whose clade union is in the subpartition of its parent.
Given a collection of history sDAG nodes $V $, we can create an edge set $E' $ containing all edges allowed by this requirement.
The resulting DAG $(V, E')$ then contains all histories that can be constructed using nodes from $V $.
If $V $ is the node set for some valid history sDAG, then the resulting DAG $(V, E') $ must also be a history sDAG.

By completing a history sDAG, additional histories are represented.
Although there is no guarantee about the weight of these new trees, it is possible that additional minimum weight trees may be found by the completed history sDAG, which makes this operation useful.

This idea is expressed in the following definition.

\begin{defn}
    Let $T $ be a collection of histories with labels in $Y $.
    Let $(V, E) $ be the history sDAG constructed from $T $.
    The \textbf{completed history sDAG constructed from} $T $ is the history sDAG $(V, E') $, where
    \begin{equation*}
        E' = \left\{(v, v') \st v, v'\in V \text{ and the clade union of $v' $ is a child clade of $v $} \right\}
    \end{equation*}
    We will also refer to $(V, E') $ as the \textbf{completion} of $(V, E) $.
\end{defn}

The completed history sDAG constructed from $T $ is a history sDAG because it includes at least those edges present in the history sDAG constructed from $T $, and all the additional edges are allowed by the definition of the history sDAG\@.
We emphasize that history sDAG completion adds no new nodes, and that a completed history sDAG is in general a much smaller object than the \emph{complete} history sDAG on a taxon set described in~\autoref{defn:complete}.

Earlier sections show that $T \subset D(T) $ for a set of histories $T $ because a history sDAG constructed from $T $ allows subhistory swaps involving conforming subhistories.
In contrast, the completed history sDAG constructed from $T $ allows any subhistories on the same leaf labels to swap, regardless of their parent nodes.

Swaps between subhistories with the same leaf label sets will not preserve history weights in the same sense as conforming subhistory swaps.
Therefore, the completed history sDAG constructed from a set of histories $T $ is not guaranteed to preserve weights in any sense.
However, the lemmas from the previous sections guarantee that any history sDAG can be trimmed to express only its minimum weight histories.
This means that the completed history sDAG can be used as a way to find even more minimum-weight histories than the original history sDAG construction, given a set of minimum-weight histories $T $.
For example, completing a history sDAG constructed from maximally parsimonious, or nearly maximally parsimonious histories, could in some cases find additional maximally parsimonious histories which wouldn't have been present before completion.

The completed history sDAG constructed from a set $T $ of histories represents all possible histories which can be constructed using the nodes of histories in $T $.
A choice of input histories can be therefore be framed as a choice of plausible pairs of labels and subpartitions, which then determines a collection of plausible histories.

\section{Exploring parsimony diversity of \sarscov clades}
The original motivation for the history sDAG was to store a collection of minimum-weight histories.
The theorems in the preceding sections show that the history sDAG is an ideal object for this task, and can discover new minimum weight histories in addition to those which we seek to store.
Because \sarscov is densely sampled relative to the rate of mutation and undergoes minimal recombination, parsimony methods are well-suited to studying its evolution~\citep{thornlow2021online}.
However, we will now demonstrate that there exists considerable uncertainty in a parsimonious reconstruction of \sarscov evolution.

Searching for maximally parsimonious trees is computationally intensive, and scales poorly as the number of leaves increases.
Traditionally, tools like \phylip's \dnapars were used to produce an assortment of maximally parsimonious trees on a given set of sequences~\citep{phylip}.
Recently, the \usher project made it possible to quickly reconstruct a single approximate parsimony tree on millions of sampled sequences~\citep{thornlow2021online}.
Neither method guarantees that the reconstructed trees are maximally parsimonious relative to all possible trees on the given leaf sequences.

Users of both methods often accept the first tree produced, ignoring the uncertainty inherent to the parsimony assumption.
However, there are in general many possible maximally parsimonious trees on a given set of leaf sequences.

Indeed, \dnapars by default outputs a non-exhaustive collection of maximally parsimonious trees.
However, for very large sets of sequences, a collection of nearly maximally parsimonious trees may be produced much more quickly using \usher.
As a demonstration, we use \usher to reconstruct trees on an assortment of \sarscov clades, extracted from the global phylogeny of public \sarscov sequences provided by the \usher project (accessed 3-3-2022)~\citep{roblanf_2020, Turakhia2021-yl}.
We allowed \usher to reconstruct trees on the set of unique sequences from each clade, as well as the ancestral sequence in the original tree, outputting a maximum of 200 trees resulting from alternative parsimonious placements of samples.
Including the ancestral sequence guarantees that the resulting reconstruction is comparable to the subtree of the global phylogeny corresponding to the same clade.
We then use the \usher utility \matoptimize \citep{MatOptimize22Ye} to attempt to optimize each tree, allowing the optimizer to make up to four moves for each sample which do not improve the parsimony score.
Allowing a few such moves is intended to increase the diversity in output trees, without requiring excessive computation time.
We saved four intermediate trees during optimization of each tree output by \usher.
Optimizations of different trees output by \usher are not guaranteed to achieve the same parsimony score.
However, even optimized trees which are not globally maximally parsimonious are likely to contain parsimony-optimal substructures.

For each clade, the collection of 800 intermediate trees resulting from these tree optimizations are used to create a history sDAG\@, after outgrouping the ancestral sequence in each.
These 800 trees are not guaranteed to be unique, and in fact there are often many duplicates.
The resulting history sDAG is then completed, trimmed to only express maximally parsimonious histories, and label-collapsed.

Whereas it is computationally expensive to construct a maximally parsimonious tree, the operations of trimming, collapsing, and completing are highly optimized, and in practice take only a few seconds for the history sDAGs used to produce~\autoref{fig:ushertrees}.
The number of operations required for the proposed trimming algorithm is bounded by $\mathcal{O}(E\cdot (MCS + MNC))$, and similarly the algorithm for completing the history sDAG is bounded by $\mathcal{O}(N^2\cdot MNC)$ where $N$ is the number of nodes, $E$ the number of edges, $MCS$ the maximum size of any set of edge descending from a node-clade pair, and $MNC$ the maximum number of child clades for any node in the history sDAG.

\begin{figure}[h!]
    \centering
    \includegraphics[width=\textwidth]{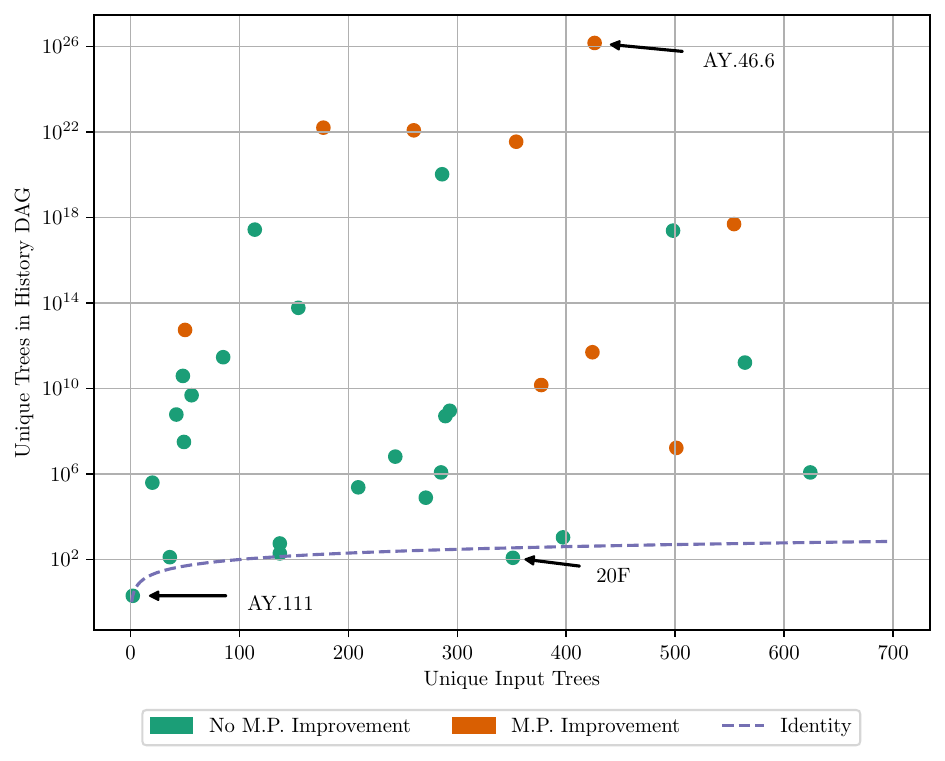}
    %WDu this comes from:
    % \includegraphics[width=\textwidth]{figures/usher_summary/usher_summary.pdf}
    \caption{\
        Unique trees found by \usher, and unique trees in the resulting history sDAG, for each selected \sarscov clade.
        Point colors indicate if the parsimony score of trees in the history sDAG is lower than the best parsimony score achieved by \usher.
        Parsimony improvement compared to \usher trees does not exceed 0.04\% for any clade.
        These data are summarized in Supplementary \autoref{table:ushertrees}
}%
\label{fig:ushertrees}
\end{figure}

The resulting history sDAG sometimes contains histories which are slightly more parsimonious than any trees found by \usher, and in most cases, the number of maximally parsimonious histories contained in the resulting history sDAG is many orders of magnitude greater than the number of histories used as input (\autoref{fig:ushertrees}).
However, this increase is far from uniform across clades.
For the clade AY.46.6, the history sDAG expresses an impressive 25 orders of magnitude more tree diversity than the input trees found by \usher, and all of those trees have a slightly better parsimony score than any tree found by \usher.
On the other hand, clade AY.111 also stands out in contrast, with only two unique trees found by \usher, and only those same two unique trees contained in the resulting history sDAG\@.

For some clades, such as 20F, the number of unique trees found by \usher is greater than the final number of trees expressed in the history sDAG\@.
Although surprising, this is not contradictory, since many of the unique trees found by \usher may have a higher parsimony score than the trees contained in the final history sDAG\@.

It is unlikely that \autoref{fig:ushertrees} reflects the true diversity of maximally parsimonious trees for each clade.
In fact, the true minimum parsimony scores for tree reconstructions of each clade may be lower than the parsimony score of trees found here.
The variation in tree diversity between clades is instead likely determined by features in the particular trees found by \usher.
Further investigation of the true diversity of maximum parsimony trees will be left for future work.

Regardless, the large diversity of trees for most clades suggests that considerable uncertainty remains about tree structure when performing a maximum-parsimony search, even after collapsing edges without mutations into multifurcations.
This uncertainty represents an opportunity to fine-tune the accepted tree in settings where parsimony is an appropriate assumption.
For example, the histories found by this method could be used as a starting point for further optimization according to criteria other than parsimony.
Such criteria, and their efficient calculation in the history sDAG, will be the subject of future work.

\section{Discussion}
This paper establishes that the history sDAG is an efficient structure for storage of similar internally labeled trees, and provides a foundation for future work to understand phylogenetic uncertainty using massive collections of parsimonious trees.

We described efficient methods for basic manipulation of the history sDAG object, and used these methods to demonstrate that for densely sampled \sarscov data, it is possible to build a history sDAG containing many alternative parsimonious evolutionary histories.
We implemented this process on clades containing up to seven thousand leaves, although it would have been feasible to use clades containing perhaps ten times as many.
Software which is currently in development will allow parsimony optimization via \matoptimize~\citep{MatOptimize22Ye} directly on the history sDAG, avoiding the time-consuming step of generating many input trees with \usher, and hopefully allowing these methods to scale to even larger datasets.

Thanks to the convenient structure of the history sDAG, it will be possible to efficiently summarize clade-level uncertainty in these histories, although such methods will be described and benchmarked in a future paper.
This approach can only be expected to work well when the tree posterior is overwhelmingly concentrated on maximally parsimonious trees, and even then clade supports estimated with the history sDAG may not be directly comparable to supports observed in a sample from the tree posterior.
However, for phylogenetic inference resulting in a single maximally parsimonious tree (which is typically arbitrarily chosen from the collection of MP trees), our method could provide a valuable understanding of the uncertainty resulting from this choice.
Clade support estimation via the history sDAG may have advantages over standard approaches to phylogenetic uncertainty estimation.
Unlike a bootstrap approach, all alternative histories in the history sDAG are built on the same data, and therefore clade support derived from the history sDAG could be more accurate for clades defined by only a few mutations~\citep{Wertheim2022-bm}.
Unlike a Bayesian approach, our method makes no attempt to fully resolve a tree when there is insufficient signal to do so, and we expect it to scale well to large data.

The history sDAG is related to various earlier works, as we now describe.

\subsubsection*{The Subsplit DAG}
The history sDAG generalizes a similar construction useful for likelihood computations and variational inference on trees, integrating out ancestral sequence uncertainty~\citep{Zhang2018-mm,Zhang2019-lw}.
Although this form of the DAG structure is not expressed in the original variational inference papers, it is described in a more recent paper~\citep{gp}.
In this \emph{subsplit DAG}, internal nodes do not contain label data, and each internal node is required to have exactly two child clades (a \emph{subsplit} is a subpartition with two parts).
That is, the subsplit DAG is a history sDAG in which internal nodes all share the same fixed label, and each node has two child clades.
The additional node label information in the history sDAG is essential for efficient storage and retrieval of maximally parsimonious trees, with the inferred ancestral sequences dictated by the parsimony assumption.

\subsubsection*{The Buneman Graph}
A construction known as the Buneman graph is related to the history sDAG\@.
In this construction, a collection of observations, each consisting of a collection of binary traits, can be arranged in a graph.
This \emph{Buneman graph} contains as subgraphs all possible maximally parsimonious trees relating the observations~\citep{semple2003phylogenetics}.
This construction has been generalized to sequences of non-binary characters~\citep{bandelt2009-ak, misra2011}, and one such generalization was applied to the problem of finding provably maximally parsimonious trees on nucleotide sequence data~\citep{misra2011}.

However, although the Buneman graph contains all maximally parsimonious trees on a set of observations, it may also contain trees which are not maximally parsimonious.
The Buneman graph is therefore not a natural data structure for storing collections of maximally parsimonious trees, since considerable additional computation may be needed to find the maximally parsimonious trees in the graph.
In contrast, the history sDAG may be trimmed to express only maximally parsimonious trees, and sampling or iterating through the trees it contains is trivial.
In addition, the history sDAG can be immediately generalized to arbitrary observed data (abstracted as node labels), and allows efficient computation and trimming with respect to weight functions other than parsimony.

\subsubsection*{Tree Fusion}
The swapping of subhistories that takes place in the history sDAG bears some resemblance to the procedure known as tree fusion, used in some parsimony software like TNT, in which clades are swapped between trees to improve parsimony scores~\citep{treefusion1999goloboff, goloboff2007divide}.

Generally, the history sDAG can be thought of as a structure which efficiently represents, and allows computation on, the set of trees resulting from all possible combinations of these clade swaps.
Thus, the history sDAG can only swap subhistories that have identical parent node labels and subpartitions.
In contrast, tree fusion can consider trees resulting from swapping any subtrees, as long as they contain the same set of samples.

Tree fusion is better approximated in the completed history sDAG, which does allow swaps of any subhistories containing the same samples.
That is, for a history sDAG $(V, E) $ constructed from a set of histories $T $, the set of histories in the completion of $(V, E) $ consists of all histories resulting from combinations of swaps involving subhistories of histories in $T $, regardless of their parent nodes.
However, subhistory swaps are still fundamentally different from the swaps of subtopologies realized during tree fusion, since subhistory swaps maintain the same ancestral node labels that were present in the original histories involved in each swap.
In order to ensure that ancestral labels are optimal in the new histories contained in the completed history sDAG, we would need an algorithm to reconstruct these ancestral states from scratch.
Such an algorithm for computing optimal ancestral states in the history sDAG would be analogous to the Sankoff algorithm for reconstructing ancestral states on trees.

Despite these limitations, \autoref{fig:ushertrees} shows that the subtree swaps which are realized in the history sDAG can be effective in reducing parsimony scores.
Although the history sDAG does not fully implement tree fusion, it concurrently applies subhistory swaps in many different histories, and allows the resulting trees to be filtered efficiently according to arbitrary criteria.
This may represent an advantage over methods which keep track of and optimize far fewer trees.

\subsubsection*{Tree Sequences}
The history sDAG also bears some similarities to the \emph{tree sequence}~\citep{Kelleher2019-tz, speidel2019}.
The tree sequence encodes a single evolutionary history for segments of a multiple sequence alignment, with changes of evolutionary history at specific points along the alignment due to recombination.
The history sDAG, on the other hand, is meant to encode an unordered collection of equally parsimonious histories.

\subsubsection*{Future Work}
We are in the process of building software that will allow us to do larger-scale inference using the history sDAG\@.
In addition to the uncertainty quantification goals described above, this software will also allow us to do broader exploration of the set of maximally parsimonious trees than previously possible.
We also hope to use the history sDAG as a means of improving MCMC sampling.

Maximally parsimonious trees may be a good starting point for inference via other methods, such as the branching process used by the tree inference package \gctree~\citep{dewitt2018using}.
To support this, we will develop efficient algorithms to make calculations on histories contained in the history sDAG\@.
We will also explore ways to search for new optimal histories, such as maximally parsimonious histories, directly within the structure of the history sDAG\@.

\section{Acknowledgements}
We thank JT McCrone and Gytis Dudas for discussions that informed this work, Mike Steel for pointing us to relevant literature, as well as Marc Suchard for suggestions on exposition.
Thanks also to Ye Cheng, Russ Corbett-Detig, Yatish Turakhia, and the rest of the \usher team for helpful discussions and their help applying \usher to the \sarscov example.
We also thank
Matthew Macaulay,
Hassan Nasif,
Anna Kooperberg,
Michael Karcher,
Tanvi Ganapathy,
Shosuke Kiami,
Seong-Hwan Jun,
Cheng Zhang,
and
Mathieu Fourment
for their work on the ``subsplit DAG,'' a closely related idea.

%WD I know this is duplicated below, but Usher asks that acknowledgement be given in the acknowledgements section.
The \sarscov data which made the exploration of diversity of parsimonious reconstructions of \sarscov clades possible is from the public databases GenBank \citep{GenBankCite}, COG-UK \citep{COGUKCite}, and the China National Center for Bioinformation \citep{CNCB1, CNCB2, CNBC3, CNBC4}.
We thank the laboratories submitting sequence data to these public databases, as well as the researchers and laboratories contributing viral samples on which these sequences are based.

\section{Declarations}

\subsection{Funding}
FAM supported by R01 AI162611; FAM is an Investigator of the Howard Hughes Medical Institute.
WSD was supported by National Institute of Allergy and Infectious Diseases Grant F31AI150163, and by a Fellowship in Understanding Dynamic and Multi-Scale Systems from the James S.\ McDonnell Foundation.
Scientific Computing Infrastructure at Fred Hutch funded by ORIP grant S10OD028685.

\subsection{Competing interests}
The authors declare no competing interests.
\subsection{Ethics approval}
No ethics approval process was required for this work.
\subsection{Availability of data and materials}
The \sarscov data used to produce clade reconstructions in the Exploring Parsimony Diversity section was read from the public \sarscov tree distributed by the \usher team at \url{http://hgdownload.soe.ucsc.edu/goldenPath/wuhCor1/UShER_SARS-CoV-2/}.
This data originates from GenBank \citep{GenBankCite} at \url{https://www.ncbi.nlm.nih.gov}, COG-UK \citep{COGUKCite} at \url{https://www.cogconsortium.uk/tools-analysis/public-data-analysis-2/}, and the China National Center for Bioinformation \citep{CNCB1, CNCB2, CNBC3, CNBC4} at \url{https://bigd.big.ac.cn/ncov/release_genome}.

\subsection{Code availability}
The history sDAG data structure described in this paper, as well as various algorithms described in this paper and in future work, are implemented in the open source Python package \historydag, which is available at \url{https://github.com/matsengrp/historydag}.

All code necessary to reproduce the \sarscov clade reconstruction example is available at \url{https://github.com/matsengrp/usher-clade-reconstructions/tree/7953eda7eb5c15556753fc23b4807b748f6a2464}.

\subsection{Authors' contributions}
Will Dumm and Frederick Matsen wrote the first draft of the manuscript, with edits and contributions to proofs from Mary Barker and edits from William DeWitt.
Will Dumm and William Howard-Snyder prepared the \sarscov clade reconstruction example.
All authors commented on previous versions, and read and approved the final manuscript.

\subsection{Open access}
This article is subject to HHMI's Open Access to Publications policy.
HHMI lab heads have previously granted a nonexclusive CC BY 4.0 license to the public and a sublicensable license to HHMI in their research articles.
Pursuant to those licenses, the author-accepted manuscript of this article can be made freely available under a CC BY 4.0 license immediately upon publication.

\bibliographystyle{sn-basic.bst}
\bibliography{main}

\appendix
\appendixpage

\section{Proofs omitted from the text}\label{app:proofs}

\reachableleaves*
\begin{proof}\label{proof:reachableleaves}
    Let $(V, E) $ be a history sDAG or subhistory on labels $Y $, and let $v\in V $ be a non-UA node.

    Let $X $ be the set of labels of leaf nodes reachable from $v $.
    $X\subset \cu(v) $ by \autoref{obs:backwardsinclusion}.
    To show inclusion in the other direction, we will show by induction on $|\cu(v)| $ that for any $\ell \in \cu(v) $, the leaf node $(\ell, \emptyset) $ is reachable from $v $ in $(V, E) $.

    As a base case, if $|\cu(v)| = 1 $, then $v $ must be the leaf node with label $\ell $, so the statement is immediately true.

    Now suppose that for any node $v'\in V $ with $|\cu(v')| < n $, and for any $\ell'\in \cu(v') $, then $(\ell', \emptyset) $ is reachable from $v' $.
    Suppose that $v = (\ell_v, U) $ is such that $|\cu(v)| = n $, and let $\ell \in \cu(v) $.
    Then $\ell \in C $ for some child clade $C\in U $.
    Since $U $ contains at least two disjoint, nonempty subsets of $Y $, it must be true that $C\subsetneq \cu(v) $.
    Any node-clade pair in a history sDAG or subhistory must have at least one descendant edge, so there exists an edge $\left(v, v_c \right) \in E $ such that $C = \cu(v_c) $, and $|\cu(v_c)| < n $.
    Since $\ell \in C $, we know that $\ell\in \cu(v_c) $, and by the inductive hypothesis, $(\ell, \emptyset) $ is reachable from $v_c $, and therefore also from $v $.
\end{proof}

\historyaltdef*
\begin{proof}
    \label{proof:historyaltdef}
    We will prove the equivalent statement that, for a history sDAG $(V, E) $ with exactly one edge descending from $\rho $, $(V, E) $ is a tree if and only if $(V, E) $ contains exactly one edge descending from each node-clade pair.
    We will prove the contrapositive of both directions.

    Assume first that there exists a node-clade pair in $(V, E) $ with at least two descendant edges.
    That is, there exist edges $(v, v_1), (v, v_2)\in E $ such that $v_1\neq v_2 $, but $\cu(v_1) = \cu(v_2) $.
    Let $\ell \in \cu(v_1) $.
    By \autoref{lemma:reachableleaves}, $(\ell, \emptyset) $ is reachable from both $v_1 $ and $v_2 $, so $(V, E) $ is not a tree.

    Now, suppose that $(V, E) $ is not a tree, meaning that there exist two edges $(v_1, v), (v_2, v) \in E $ with the same child node.
    Since all nodes in a history sDAG must be reachable from $\rho $, there exist paths in $E $ connecting $\rho $ to both $v_1 $ and $v_2 $.
    $(V, E) $ has only one edge exiting $\rho $, so these two paths must diverge at some non-UA node $v_r = (\ell_r, U_r) \in V $.
    That is, there are edges $(v_r, v_1'), (v_r, v_2') \in E $ such that $v_1 $ is reachable from $v_1' $ and $v_2 $ is reachable from $v_2' $.
    Therefore, $v $ is reachable from both $v_1' $ and $v_2' $, so $\cu(v) \subset \cu(v_1') $ and $\cu(v) \subset \cu(v_2') $, and $\cu(v_1') \cap \cu(v_2') \neq \emptyset $.
    However, since $v_1' $ and $v_2' $ are children of the node $v_r $, both $\cu(v_1') $ and $\cu(v_2') $ must be elements of $U_r $.
    Elements of $U_r $ are disjoint, nonempty subsets of $Y $, so $\cu(v_1') = \cu(v_2') $.
    We have demonstrated that the two edges $(v_r, v_1') $ and $(v_r, v_2') $ descend from the same node-clade pair $(v_r, \cu(v_1')) $ in $(V, E) $.
\end{proof}

\acyclic*
\begin{proof}\label{proof:acyclic}
    We will show that no edge may take part in a cycle.

    Recall first that the UA node only admits outgoing edges, so no edge exiting $\rho $ can be part of a cycle.

    Consider an edge $e = \left(v_p = (l_p, U_p), v = (l, U) \right) $ whose parent is not $\rho $.
    If $v_p $ is not the UA node, then $|U_p| \geq 2 $, and either $U = \emptyset $ or $\bigcup\limits_{C\in U} \in U_p $.
    In the first case, $v $ is a leaf node, which can only accept incoming edges, so the edge $e $ cannot be part of any cycles.
    In the second case, since $|U_p| \geq 2 $ and elements of $U_p $ are nonempty, disjoint subsets of $X $,
    \[\left | \bigcup\limits_{C\in U_p} C \right | > \left | \bigcup\limits_{C\in U} C \right |. \]
    The same inequality is true of any edge reachable from $v $ which does not terminate at a leaf node, so no edge reachable from $v $ can have $v_p $ as a target.
\end{proof}

\nodesareroots*
\begin{proof}\label{proof:nodesareroots}
    We will prove this by induction on $|\cu(v)| $.

    As the base case, suppose $|\cu(v)| = 1 $.
    Then $v $ must be a leaf node, and the subhistory we seek is the one consisting of only the node $v $.

    Now suppose it's true that for any node $v $ with $|\cu(v)| < n $, there's a subhistory $s $ rooted at $v $ in $(V, E) $.

    Let $v = (\ell, U) \in V $, and suppose that $|\cu(v)| = n $.
    Then $U = \left\{C_1, \ldots, C_m \right\} $ where $C_1, \ldots, C_m $ are $m \geq 2 $ disjoint subsets of $Y $.
    For each $C_i \in U $, by the definition of the history sDAG, there exists at least one edge $(v, v_i) \in E $ with $\cu(v_i) = C_i $.
    Also, since $|C_i| < n $, each $v_i $ is guaranteed to have a subhistory $s_i = (V_{s_i}, E_{s_i}) $ in $(V, E) $, with $v_i $ as its root, by the inductive hypothesis.

    Notice that the node sets $V_{s_i} $ for $1\leq i\leq m $ are pairwise disjoint:
    Let $v' \in V_{s_k} $ and $v''\in V_{s_j} $, for $k\neq j $, $1\leq k, j \leq m $.
    A necessary condition for node equality is that $\cu(v') = \cu(v'') $.
    But notice that $\cu(v') \subset \cu(v_k) = C_k $, and $\cu(v'') \subset \cu(v_j) = C_j $.
    Since $C_j\cap C_k = \emptyset $, $\cu(v') \neq \cu(v'') $, so $v' \neq v'' $, and $V_{s_j} \cap V_{s_k} = \emptyset $.

    Therefore, we can build a subhistory rooted at $v $ consisting of $v $, the subtrees $s_i $ for $1\leq i\leq m $, and the edges connecting $v $ to the root node of $v_i $ for each subtree $s_i $.
    That is,
    \begin{equation*}
        s = \left(\left [\left\{v \right\} \cup \bigcup_{i=1}^m V_{s_i}\right ], \left [ \left\{(v, v_i) \mid 1\leq i\leq m \right\} \cup \bigcup_{i=1}^m E_{s_i}\right ] \right)
    \end{equation*}
    is the subhistory we seek, rooted at $v $.

    This subhistory has exactly one edge descending from each node-clade pair because $s_i $ are subhistories, and because for each child clade $C_i $ of $v $, the edge $(v, v_i) $ descends from $(v, C_i) $.
\end{proof}

\dagisitshistories*
\begin{proof}
    We will argue that for any edge $e = (v, v_c)\in E $, there exists a history $(V', E') $ in $(V, E) $ with $e\in E' $.

    Let $\left(e_i=(v_i, v_{i+1}) \right)_{i=0}^{n-1} $ be a sequence of edges in $E $ which is a path from $\rho $ to $v_c $, so that $v_0 = \rho $ and $e_{n-1} = (v_{n-1}, v_n) = (v, v_c) = e $.
    For $1\leq i\leq n $, let $s_i $ be a subhistory rooted at $v_i $, which exists by \autoref{lemma:nodesareroots}.

    Now recursively construct $s_i' $ for $1\leq i< n $ by replacing the edge descending from the node-clade pair $\left(v_i, \cu(v_{i+1}) \right) $ in $s_i $, and the subhistory consisting of all nodes and edges reachable from the child node of that edge, with the edge $e_i $ and the subhistory $s_{i+1} $, rooted at $v_{i+1} $.
    Finally, let $s_n' = s_n $.

    $s_i' $ remains rooted at $v_i $, and now contains the edge $e_i $ and all the edges $e_j $, for $i < j < n $, including the edge $e $.

    $s_1' $ is then a subhistory rooted at $v_1 $ which contains $e $.
    \autoref{lemma:historyaltdef} implies that by adding the edge $(v_0=\rho, v_1) $ to $s_1' $, we've constructed a history in $(V, E) $ containing $e $.
    Note that in the language of following sections, this history can also be expressed as $t\tl s_2 \tl \cdots \tl s_n $, where $t $ is the history consisting of $s_1 $ and the edge $(\rho, v_1) $.

    To finish the proof, let $(V_T, E_T) $ be the history sDAG constructed from $T $.
    Since a history sDAG must be connected, to show that $(V, E) $ and $(V_T, E_T) $ are equal is to show that $E_T = E $.
    Any edge in $E_T $ must be present in some history in $T $, and since $T $ is the set of histories in $(V, E) $, any such edge must be in $E $.
    Therefore, $E_T \subset E $.
    Also, we just showed that any edge $e\in E $ must take part in some history in $T $, and therefore $e $ must also be in $E_T $.
    Therefore, $E_T = E $.
\end{proof}

%WDu conforming subtree swap lemmas start here:
\subsection{History Weights}
The lemmas in this appendix subsection are necessary for the proof of \autoref{thm:updown}.
The proof of that theorem is given at the end of this subsection.

\begin{defn}
    Let $s $ and $s' $ be subhistories of histories $t $ and $t' $ in some ambient history sDAG\@.
    We say that $s $ and $s' $ are \textbf{conforming} subhistories if $s' $ has the same set of leaf labels as $s $, and the parent node in $t' $ of $s' $ is the same as the parent node of $s $ in $t $.

    More formally, if $t = (V, E) $ and $t' = (V', E') $, and $s = (V_s, E_s), s' = (V_s', E_s') $, with $V_s \subset V $ and $V_s' \subset V' $, then $s $ and $s' $ are conforming if:
            \begin{enumerate}
                \item $v_p = v_p'$ for $v_p, v_p' $ the parent nodes of $s $ and $s' $ in $t $ and $t' $, respectively
                \item $L(s) = L(s')$.
            \end{enumerate}
\end{defn}

Notice that since no internal node in a history may have exactly one child, a history may not contain two distinct subhistories with the same leaf nodes.
Therefore, given a history $t $ and a subhistory $s' $, a choice of subhistory $s $ of $t $ conforming with $s' $ is guaranteed to be unique, if it exists.

Notice also that the definition of conforming subhistory does not allow a history to be conforming with any subhistory of itself.
%WDu This is why the definition of subhistory is strict, not allowing
%a subhistory to contain \rho
To evaluate conformity of a subhistory, there must be some ambient parent node (\autoref{fig:conforming_subtrees_example}).

\begin{figure}[H]

    \centering
    \includegraphics{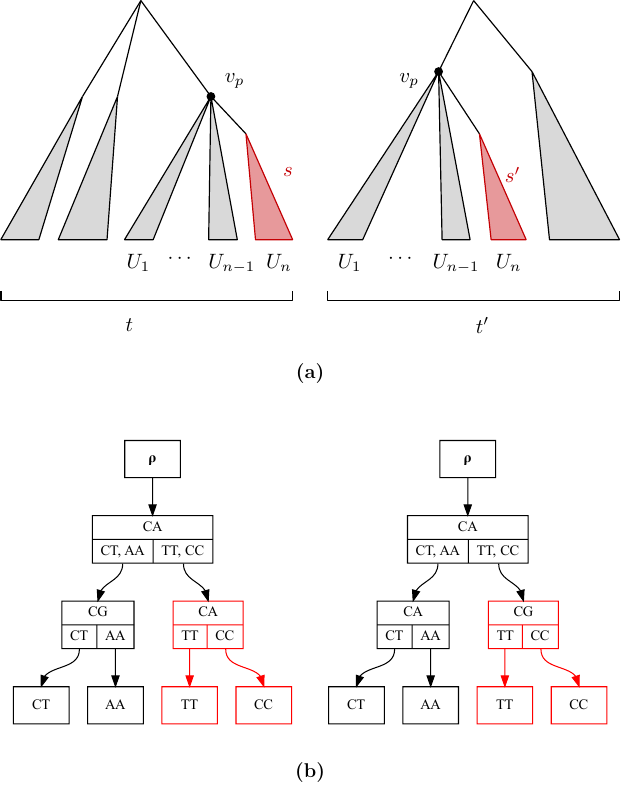}
\caption{\
    In both (a) and (b), the two subhistories highlighted in red are conforming, since they share the same set of leaves, and in each respective history their parent node has the same label and descendant clades.
We include this example to emphasize that conforming subhistories need not have the same internal node labels
}%
\label{fig:conforming_subtrees_example}
\end{figure}

We can now define the exchange of substructures that takes place between histories in the history sDAG\@.
\begin{defn}\label{defn:subhistoryswap}
    Let $t = (V, E)$ be a history, and let $s = (V_s, E_s) $ be a subhistory of $t $.
    Also, let $(V_d, E_d) $ be a history sDAG, and let $s' = (V', E') $ be any subhistory of $(V_d, E_d) $ conforming with $s $.

    A \textbf{subhistory swap} of $t $ and $s' $ is a history with the structure and labeling of $t $, except that the subhistory $s $ of $t $ is replaced with the subhistory $s' $.

    More formally, the subhistory swap replacing $s $ with $s'$ is the history with nodes $(V\setminus V_s)\cup V' $, and edges
            \begin{equation*}
                \left (E\setminus \left ( E_s \cup \left\{(v_p, v) \right\} \right ) \right ) \cup \left (E' \cup \left\{(v_p, v') \right\}\right )
            \end{equation*}
         where $v_p $ is the parent node of $s $ in $t $, $v $ is the root node of $s $, and $v' $ is the root node of $s' $.
\end{defn}

\begin{defn}
    Let the \textbf{swap operator} $\tl $ be a left-associative operator on history, subhistory pairs, defined so that $t\tl s' $ is the subhistory swap of $t $ and $s' $, if a subhistory of $t $ conforming with $s' $ exists.
    $t\tl s' $ is undefined if no such subhistory exists.
\end{defn}

Notice again that the subhistory (right argument of $\tl $) in a subhistory swap must exist in the context of some ambient history sDAG, so that it can be evaluated whether swapped subhistories are conforming.

The definition of conformity is slightly more restrictive than it needs to be to guarantee that subhistory swaps of conforming subhistories preserve parsimony.
In particular, there is no need to require that the parent nodes of the swapped subhistories have the same subpartitions.
However, this assumption is natural in the context of the history sDAG structure, and is necessary for the argument to extend to edge weight functions that depend on nodes' subpartitions.

\begin{lemma}
    \label{lemma:tlwelldefined}%
The operator $\tl $ is well-defined on subhistories.
Also, given subhistories $t, s' $ both with labels in $Y $, and with the leaves of $t $ labeled by $X\subset Y $, then $t\tl s' $ is a history with labels in $Y $ and leaves labeled by $X $.
That is, $\tl $ preserves the leaf labels of its left argument.
\end{lemma}
\begin{proof}
To show that $\tl $ is well-defined, we need to show that given histories $t, s' $, the subhistory swap $t\tl s' $ is a history, and is uniquely determined by the choice of $t $ and $s' $.
$t\tl s' $ is a history directly from the definition, and by the observation that since neither $t $ nor $s' $ may have unifurcations, their subhistory swap may not either.
$t\tl s' $ replaces a subhistory $s $ of $t $ with $s' $, where $s $ must have exactly the same leaf label set as $s' $.
If such a choice of $s $ exists, it must be unique by the assumption that nodes in a history may not have exactly one child.
This guarantees that no two nodes in a history are above the same set of leaves.

Now assume that $t $ and $s' $ are subhistories on labels $Y $, and $t $ has leaves labeled by $X\subset Y $.
To see that $t\tl s' $ is a history with labels in $Y $ and leaves labeled by $X $, notice first that $s' $ must have nodes labeled bijectively by a set $C\subset X $, the same set of leaf labels as the subhistory in $t $ that $s' $ replaces.
Therefore the labeling on $t\tl s' $, restricted to leaf nodes, is bijective as a union of two bijective functions with disjoint domains, and images partitioning $X $.
The labeling on $t\tl s' $ maps into $Y $ as a union of functions which both map into $Y $.
\end{proof}

We now describe the sense in which subhistory swaps preserve history weight.
\begin{lemma}
    \label{lemma:updownlittle}
    Let $t_1 = (V_1, E_1) $ and $t_2 = (V_2, E_2) $ be histories on labels $Y $.
    For $i\in \left\{1,2 \right\} $, let $s_i $ be a subhistory of $t_i $, so that $s_1 $ and $s_2 $ are conforming.

    Let $t_1' = t_1 \tl s_2 $ be the history constructed by replacing $s_1 $ with $s_2 $ in $t_1 $, and similarly define $t_2' = t_2 \tl s_1$ to be the history constructed by replacing $s_2 $ with $s_1 $ in $t_2 $.
    Finally, suppose that $f $ is an edge-weight function taking values in a weight set $W $, clade-ordered with respect to $Y $.
    Then $g_f(t_1') < g_f(t_1) $ if and only if $g_f(t_2') > g_f(t_2) $.

\end{lemma}
\begin{proof}
    Let $v_i $ be the parent node of $s_i $ in $t_i $, and let $K_i = g_f(s_i^{v_i})$ for $i\in \left\{1,2\right\}$.
    That is, $K_i $ is the weight of the augmented subhistory $s_i$ and its parent edge.
    Then for some weights $w_1, w_2\in W $, $g_f(t_i) = w_i + K_i $, and also $g_f(t_1') = w_1 + K_2 $ and $g_f(t_2') = w_2 + K_1 $.
    The following are equivalent, since $K_i $ are weights of subhistories below the same clade, and $W $ is clade-ordered:
    \begin{align*}
        w_1 + K_2 &< w_1 + K_1\\
        K_2 &< K_1\\
        w_2 + K_2 &< w_2 + K_1\\
    \end{align*}
    so $g_f(t_1') < g_f(t_2) $ if and only if $g_f(t_2') > g_f(t_2) $.
\end{proof}

To extend the conclusion of this lemma to all the histories in the history sDAG, we need a few more lemmas:

\begin{lemma}\label{lemma:closedundertl}
    Suppose $t_1, \ldots, t_n $ are histories in the history sDAG $(V, E) $, and $s_i $ is a subhistory of $t_i $ for $2\leq i\leq n $.
    Then $t_1 \tl s_2 \tl \ldots \tl s_n $ is a history in $(V, E) $.
\end{lemma}
\begin{proof}
    We need only show this is true for $n=2 $, since subhistory swaps are left-associative.
    Let $t_1, t_2 $ be histories in the history sDAG, and let $s_2 $ be a subhistory of $t_2 $, conforming with some subhistory $s_1 $ of $t_1 $.
    Also let $v_1, v_2 $ be the root nodes of $s_1 $ and $s_2 $ respectively.
    Conformity means that the parent node $v_p $ of $s_1 $ in $t_1 $ is the same as the parent node of $s_2 $ in $v_2 $.
    Notice that all the edges in $t_1 $ are in $E $, as well as all the edges of $s_2 $, since we assumed that $t_1, t_2 $ are histories in the DAG\@.
    Also notice that the edge $(v_p, v_2) $ is in $E $, because it is an edge in $s_2 $.
    Therefore, all the edges in $t_1\tl s_2 $ are in $E $, and $t_1\tl s_2 $ is a history in the history sDAG\@.
\end{proof}

The following lemma describes how any history in a history sDAG built from a collection of histories $T$ can be built from a collection of swaps operating on subhistories from $T$.
\begin{lemma}\label{lemma:whataredaghistories}
    Let $t\in D(T) $ be a history in the history sDAG $(V, E) $ constructed from a collection of histories $T $.
    Then for some sequence of histories $(t_i)_{i=1}^n $ in $T $, and choices of subhistories $s_i $ of $t_i $ for all $i$, $t = t_1 \tl s_1 \tl \ldots \tl s_n $.
\end{lemma}
\begin{proof}
    Let $t = (V_t, E_t) $ be a history in $(V, E) $.
    Every edge in $E_t $ must appear in some $t'\in T $.
    Since $t $ is a tree, there exists a preordering $(v_i)_{i=0}^n $ of vertices in $V_t $ so that $v_j $ is reachable from $v_i $ only if $i \leq j $.
    That is, if $i > j $, then $v_j $ must not be reachable from $v_i $.
    Let $(e_i)_{i=1}^n $ be an ordering of edges in $E_t $ such that $v_i$ is the target node of $e_i $ for all $1\leq i\leq n $.
    We will use the notation $v_i' $ to denote the parent node of the edge $e_i $.
    Notice then that for edges $e_i = (v_i', v_i) $ and $e_j = (v_j', v_j) $, if $v_i = v_j' $, then $i < j $.

    Finally, also define a sequence of parent node-clade pairs $(p_i)_{i=1}^n $ so that $p_i = (v_i', \cu(v_i)) $.
    We will say that an edge $e_j=(v_j', v_j) $ is reachable from a node-clade pair $p_i $ if there exists a path of edges ending with $e_j $ such that the first edge in the path descends from the node-clade pair $p_i $.
    Notice that since the sequence $(e_i) $ is a preordering of the history $t $, and since only one edge may descend from each node-clade pair in a history, if $e_j $ is reachable from the node-clade pair $p_i $, then $i \leq j $.

    Now choose a sequence $(t_i)_{i=1}^n $ of histories in $T $ so that $e_i $ is an edge in $t_i $ for all $i $, and let $s_i $ be the subhistory of $t_i $ rooted at $v_i $, the child node of the edge $e_i $.
    The edge $e_i $ is not reachable from the parent node-clade pair of any $s_k$ with $k > i $, because the indices are chosen to preorder nodes and edges.
    Notice that a subhistory swap can only change edges reachable from the shared parent node-clade pair of the subhistories being swapped.
    Assume temporarily that $e_i $ is in $t_1\tl s_1\tl \cdots \tl s_i $.
    Then the edge $e_i $ must be in the history $t_1\tl s_1\tl \cdots \tl s_k $ for $k > i $, since $e_i $ must not be reachable from $p_k $.

    Because of this, to show that all edges in $E_t $ are in $t_1\tl s_1\tl \cdots \tl s_n $, we need only show that the edge $e_i $ is in $t_1\tl s_1 \tl \cdots \tl s_i $ for all $1\leq i\leq n $, which we now establish.
    Inducting on $i $, notice first that $e_1 $ is in $t_1\tl s_1 = t_1$, by our choice of $t_1 $.
    Supposing that $e_j $ is in $t_1\tl s_1\tl \cdots \tl s_j $ for all $j < i $, notice that $v_i' = v_j$ for some $j < i $, so that $v_i' $ is in $t_1\tl s_1\tl\cdots\tl s_{i-1} $, and $s_i $ is conforming with some subhistory of $t_1\tl s_1\tl\cdots\tl s_{i-1} $.
    The subhistory swap with $s_i $ therefore replaces the unique child node $v $ in $t_1\tl s_1\tl\cdots\tl s_{i-1} $ which descends from the node-clade pair $(v_i', \cu(v_i)) $, and all of its descendants, with $s_i $, which is rooted at $v_i $ and attached below $v_i' $ with the edge $(v_i', v_i)= e_i $ in $t_1\tl s_1\tl \cdots\tl s_i $.

    Therefore, $t_1\tl s_1\tl \cdots\tl s_n $ contains at least all those edges in $E_t $.
    Furthermore, $t_1\tl s_1\tl \cdots\tl s_n$ is a history with the same leaves as $t $, so it can't contain any more edges than those in $E_t $ and remain a tree.
    That is, $t = t_1\tl s_1\tl \cdots \tl s_n $.
\end{proof}

With the preceding lemmas, it is finally possible to prove the main result of this section:

\updownthm*

\begin{proof}
    \label{proof:updownthm}
    By \autoref{lemma:whataredaghistories}, any history $t\in D(T) $ can be expressed as a finite sequence of subhistory swaps involving histories in $T $.
    We will induct on $n $, the number of subhistory swaps involving histories in $T $ required to express $t $.
    First, suppose the history $t $ can be expressed as $t = t_1\tl s_2 $, a subhistory swap involving histories $t_1, t_2\in T $, and the subhistory $s_2 $ of $t_2 $, conforming with a subhistory $s_1 $ of $t_1 $.
    Then by \autoref{lemma:updownlittle}, $g_f(t) < K $ if and only if $g_f(t_2\tl s_1) > K $.
    $t_2\tl s_1 \in D(T) $ by \autoref{lemma:closedundertl}, so we've shown that $g_f(t) < K $ implies there exists a history $t'\in D(T) $ with $g_f(t') > K $.
    By the same argument, if $g_f(t) > K $, then there exists a history $t'\in D(T) $ with $g_f(t') < K $.

    Now suppose that for $i < n $, and for any $t\in D(T) $ which can be expressed as $t= t_1\tl s_2\tl \cdots \tl s_i $ for $s_i $ subhistories of histories in $T $,
    \begin{itemize}
        \item if $g_f(t) < K $ then there exists a history $t' \in D(T)$  with $g_f(t') > K $, and
        \item if $g_f(t) > K $ then there exists a history $t' \in D(T)$  with $g_f(t') < K $.
    \end{itemize}

    Let $t\in D(T) $ be expressible as $t = t_1\tl s_2\tl \cdots \tl s_n $, where $s_2, \ldots ,s_n $ are subhistories of histories $t_2, \ldots, t_n $, and $t_1, \ldots, t_n \in T $.
    Suppose $g_f(t) < K $, and let $t_* = t_1\tl s_2\tl\cdots\tl s_{n-1} $.
    Notice that $t_* \in D(T) $ by \autoref{lemma:closedundertl} since $t_1, \ldots, t_n \in D(T) $.
    We seek to show there exists $t'\in D(T) $ with $g_f(t') > K $.

    If $g_f(t_*) > K $, then $t_* $ is the history we seek.

    If $g_f(t_*) < K $, then the history we seek exists by the inductive hypothesis, since $t_* $ is expressible as a subhistory swap involving $n-1 $ histories in $T $.

    If $g_f(t_*) = K $, let $s $ be the subhistory of $t_* $ conforming with $s_n $.
    Notice by \autoref{lemma:closedundertl}, $t_n \tl s \in D(T) $.
    By \autoref{lemma:updownlittle}, $g(t_n \tl s) > K$.

    A similar argument shows that if $g_f(t) > K $, there exists $t'\in D(T) $ with $g_f(t') < K $.
\end{proof}

%WDu conforming subtree swap lemmas end here.

\subsubsection{Trimming the history sDAG}
\trimistrim*
\begin{proof}
    \label{proof:trimistrim}
    First, $(\underline V, \underline E) $ is a history sDAG: $\underline V\subset V $, $\underline E \subset E $, and all nodes in $\underline V $ are reachable from $\rho $ by construction.
    Also, for each node $v = (\ell, U) \in \underline V $, and each $C\in U $, there is at least one edge descending from the node-clade pair $(v, C) $, since at least one edge must achieve the minimum augmented subhistory weight in each clade.

    Since a history sDAG is uniquely determined by the histories it contains by \autoref{lemma:dagisitshistories}, it's enough to show that these two history sDAGs contain the same set of histories.
    First, let $t $ be a history in $(V', E') $, and let $e = (v, v_c) $ be any edge in $t $.
    Since $t $ achieves the minimum weight of any history in $(V, E) $, $M_f(v_c) + f(v, v_c) $ must be equal to $M_f(v, \cu(v_c)) $.
    If this were not true, there would necessarily exist some subhistory $s $ in $(V, E) $ for which $g_f(t\tl s) < g_f(t) $, and $t\tl s $ would be a history in $(V, E) $, contradicting the assumption that $t $ is a minimum-weight history in $(V, E) $.
    Therefore, all edges in $t $ are in $\underline E' $.
    All nodes in $t $ are reachable from $\rho $ via paths in $\underline E' $, in particular the paths which follow the history $t $, so all nodes in $t $ are in $\underline V $, and all edges in $t $ are in $\underline E $.
    That is, $(\underline V, \underline E) $ contains at least all minimum-weight histories in $(V, E) $.

    To show that $(\underline V, \underline E) $ contains only minimum-weight histories, let $K $ be the minimum weight of any history in $(V, E) $, and let $t $ be a history in $(\underline V, \underline E) $ with $g_f(t) > K $.
    Consider the edge $e = (v, v_c) $ in $t $ closest to $\rho $ which is not in $(V', E') $.
    Since $e\notin E' $, $e $ must not be an edge in any minimum-weight history in $(V, E) $.
    However, since $e $ is the closest edge to $\rho $ in $t $ which is not in $(V', E') $, it must be true that $v\in V' $.
    That is, there must be no subhistory $s\in \tgts(v_c) $ in $(V, E) $ such that $g_f(s^v) = M_f(v, \cu(v_c)) $, and the edge $e $ must not be in $\underline E $.
    Therefore, $t $ is not in $(\underline V, \underline E) $, and $(\underline V, \underline E) $ contains exactly the minimum-weight histories in $(V, E) $, so by \autoref{lemma:dagisitshistories}, $(\underline V, \underline E) = (V', E') $.
\end{proof}
\subsubsection{Collapsing histories}
\collapseone*
\begin{proof}
    \label{proof:collapseone}
    Let $(V_!, E_!) $ be the DAG constructed from $T' $.
    We must show that $E' = E_! $, so that by construction, $V' = V_! $.

\item First, to show that $E'\subset E_! $, let $(v_1, v_2)\in E' $.
    \begin{itemize}
        \item If $\left\{v_1, v_2 \right\}\cap \left\{v_p, v_p', v_c \right\} = \emptyset $, then $(v_1, v_2) \in E_! $ because collapsing in histories only modifies edges incident to the edge being collapsed.

        \item If $v_2 = v_p $, then $v_p $ was not removed from $V $, meaning that some edge $(v_p, v) \in E $ must exist, with $\cu(v) = \cu(v_c) $.
            $T $ contains all the histories in $(V, E) $, so there is a history $(V_t, E_t) $ in $T $ so that $(v_p, v) \in E_t $.
            Since each history has exactly one edge descending from each node-clade pair, $(v_p, v_c)\notin E_t $, and $(V_t, E_t) \in T' $.
            Therefore, $(v_p, v) \in E_! $.
        \item If $v_2 = v_p' $, then $(v_1, v_2) \notin E $, but $(v_1, v_p) \in E $, and $(v_p, v_c) \in E $.
          Since $(v_1, v_p) $ and $(v_p, v_c) $ are adjacent in $(V, E) $, there's a subhistory which contains both edges, and since $(v_p, v_c)$ is $b $-collapsible and $T$ is a $b $-collapsible edge cover, there is a history $(V_t, E_t)\in T $ which contains both the subhistory, and consequently, both edges. The corresponding label-collapsed history in $T' $ contains $(v_1, v_2) $.
            Therefore, $(v_1, v_2) \in E_! $.
        \item If $v_2 = v_c $, then $v_1\neq v_p $, and $(v_1, v_2) \in E $.
            Some history in $T $ must contain the edge $(v_1, v_2) $, and may not contain the edge $(v_p, v_c) = (v_p, v_2) $, in order to be a tree.
            Therefore, this history is unchanged in $T' $, and $(v_1, v_2) \in E_! $.
        \item If $v_1 = v_p $, then $v_2\neq v_c $ since $(v_p, v_c) \notin E' $.
            Therefore, $(v_1, v_2) \in E $, and by the same reasoning as above, $(v_1, v_2) \in E_! $.
        \item If $v_1 = v_c $, then again $(v_1, v_2) \in E $, so $(v_1, v_2) \in E $.
        \item If $v_1 = v_p' $, then either $(v_c, v_2) \in E $ or $(v_p, v_2) \in E$.
            There exists a history $(V_t, E_t) \in T $ with $(v_p, v_c) \in E_t $, and either $(v_c, v_2) $ or $(v_p, v_2) $ in $E_t $, in which collapsing $(v_p, v_c) $ yields the edge $(v_p', v_2) $.
            The resulting history is in $T' $, so $(v_1, v_2) \in E_! $.
    \end{itemize}
    We've addressed all the situations where one of $v_1, v_2 \in \left\{v_p, v_p', v_c \right\} $.
    Both nodes can't be in that set, because no pair of nodes  in $\left\{v_p, v_p', v_c \right\} $ can have an edge between them in $E' $, by construction.

\item Now, to show that $E_!\subset E' $, let $(v_1, v_2)\in E_! $.
    First, notice that $E_!\subset E^+ $, because $E^+ $ contains all the edges that are added to histories in $T $ when collapsing $(v_p, v_c) $, and $E_! $ does not contain $(v_p, v_c) $.

    By definition, $(v_1, v_2)\in E_t $, for some $(V_t, E_t)\in T' $.
    If $E^- = E^+ $, and since $(v_1, v_2) \in E_t $, $v_1 $ is reachable from the UA node, and $(v_1, v_2) \in E' $.
    If $E^-\neq E^+ $, then $v_p $ must have had no edges descending from its child clade $\cu(v_c) $, in $E^+ $.
    That means $v_p\notin V_t $, since a history must have exactly one descendant edge for each node-clade pair.
    Therefore, no removed parent edges are in $E_t $, and $E_t\subset E^- $.
    This means that $v_1 $ is reachable from the UA node in $E^-$, and $(v_1, v_2) \in E' $.

    Therefore, $E' = E_! $ and $V' = V_! $.
\end{proof}
\collapsingterminates*
\begin{proof}\label{proof:collapsingterminates}
    Let $T_0 $ be a label-collapsible edge cover of $(V_0, E_0) $, and define a sequence of sets of histories $(T_i)_{i\in \N} $, where $T_k = T_{k-1} $ if $(V_k, E_k) = (V_{k-1}, E_{k-1}) $, and otherwise let $T_k$ be obtained by collapsing all histories in $T_{k-1}$ at the edge $e_{k-1}$.
    Notice that if such an $N $ exists, then $T_0' \subseteq T_N $ by \autoref{lemma:collapseone}.
    Therefore we need only show that such an $N $ exists.

    %We will show that for a given history sDAG $(V, E)$, the maximum number of collapsing edge iterations required to fully collapse the DAG is finite.
    Let $T$ denote a label-collapsible edge cover of $(V, E)$, and denote the multiset of collapsible edges in all $t\in T$ as $E_{col}$.
    For each collapsible edge $e\in E$, the label-collapsed DAG $(V', E')$ is equivalent to the DAG obtained from the label-collapsed histories $T'$ by \autoref{lemma:collapseone}.

    Note that the number of trees in $T'$ is equal to the number of trees in $T$. However, the total number of unique trees in $T'$ can be smaller than in $T$ since collapsing an edge in two different trees can produce the same resulting tree.
    Also, since collapsing an edge in a history does not introduce any new edges in that history, collapsing $e$ strictly reduces the number of edges in $E_{col}$.
    So the multiset of collapsible edges in $T'$ is a strict subset of $E_{col}$.
    We will demonstrate that $T'$ is a label-collapsible edge cover of $(V', E')$.
    Since the set $E_{col}$ is finite, these results imply that any such sequence of history sDAGs results in a label-collapsed DAG in a finite number of steps.

    To show that $T'$ is a label-collapsible edge cover of $(V', E')$, we show that for a collapsible edge $e_c\in E'$, every subhistory in $(V', E')$ which contains $e_c$ is contained in $T'$.

    Let $e_c$ be given, and suppose $s'$ is any subhistory in $(V', E')$ containing $e_c$.

    If every edge in $s'$ is disjoint from the vertices $\{v_p', v_p\}$, then there is an identical subhistory $s$ in $(V, E)$, and, by the label-collapsible edge covering property of $T$, there exists $t\in T$ containing $s$.
    Since every edge in $s$ is disjoint from the set of edges altered by collapsing at $e$, collapsing $t$ at $e$ yields a history in $T'$ that contains $s=s'$ as a subhistory.

    If there is an edge in $s'$ of the form $(v, v_p')$ then consider the corresponding un-collapsed subhistory $s$ in $(V, E)$ consisting of edges
    \begin{align}
      s &=\left\{(v, v_p)\right\}\nonumber\\
             &\cup \left\{(v_p, v_n) \st (v_p', v_n) \in s', \cu(v_n)\neq\cu(v_c) \right\}\nonumber\\
             &\cup \left\{e\right\}\nonumber\\
             &\cup \left\{(v_c, v_n)\st (v_p', v_n)\in s', \cu(v_n)\subset\cu(v_c)\}\right\}\nonumber\\
             &\cup \left\{(v_1, v_2)\st (v_1, v_2)\in s', \{v_1, v_2\}\cap \{v_p'\}=\emptyset\right\}\nonumber
    \end{align}
    Where the edges adjacent to $v_p'$ are replaced with the corresponding structures in $E$.
    By construction, $s$ is a subhistory in $(V, E)$ containing a collapsible edge $e$ and such that collapsing at $e$ yields $s'$.
    Since $T$ is a label-collapsible edge cover for $(V, E)$, there exists $t\in T$ containing $s$, and, since collapsing $s$ at $e$ yields $s'$, the label-collapsed history $t'\in T'$ contains $s'$.
    The analogous argument holds if $s'$ is a subhistory containing an edge of the form $(v_p', v)$.

    If there is an edge in $s'$ of the form $(v_p, v)$, then by the observation following \autoref{lemma:collapseone}, this implies that there is another edge descending from the node-clade pair $(v_p, \cu(e))$ distinct from $e$, and that $s'$ can be viewed as a subhistory of a subhistory containing that alternative edge.
    So the subhistory $s'$ corresponds to a subhistory $s$ in $(V, E)$ which belongs to a history $t$ that cannot contain $e$.
    Since $s$ is contained in a history that does not contain $e$, collapsing at $e$ does not affect $s$, and so collapsing $s$ yields $s'=s$.
    Thus $s$ is a subhistory in $(V, E)$ that contains the collapsible edge $e_c$, and hence there exists a history $t\in T$ containing $s$.
    Since $t$ does not contain $e$, collapsing at $e$ yields $t'\in T'$ which, trivially, contains $s'=s$.

    And so $T'$ is a label-collapsible edge cover for $(V', E')$.
\end{proof}

% \clearpage
% \section*{Supplementary Materials}
% \beginsupplement
% Supplementary text and figures here.

\subsection{Histories are labeled trees:}\label{sec:correspondence}
In this subsection, we show that history substructures in the history sDAG are in bijection with isomorphism classes of rooted, internally labeled, multifurcating trees.
There will be a number of notational differences from the rest of the paper.
Rather than a history, $t $ will denote a labeled tree, and $s $ a subtree of a labeled tree.
$\tau $ will denote a tree's graph structure, in which nodes are abstract objects rather than the label, subpartition pairs that the history sDAG consists of.
The function $L $ will denote the set of leaf nodes below an internal node in a labeled tree.
Also, $\varphi $ will denote a labeling function of a labeled tree, rather than a disambiguation of a history.

$Y $ will continue to mean a set of labels, as in the rest of the paper.

\begin{defn}
    An \textbf{(internally) labeled tree} $t = (\tau, \varphi) $ is
    \begin{itemize}
        \item a rooted, multifurcating tree $\tau=(V, E) $, and
        \item a labeling function on vertices, $\varphi:V\to Y $, where $Y $ is a label set.
    \end{itemize}
    We will let $L(t) $ refer to the set of leaf nodes of the tree $\tau $, and require that
    \begin{itemize}
        \item no node in $\tau $ has exactly one child, and
        \item the labels on leaf vertices must be unique (that is, $\left . \varphi \right |_{L(t)} $ must be injective), but labels on internal vertices need not be (that is, $\varphi $ need not be injective or surjective).\\
    \end{itemize}

    % A history $s = (\tau_s, \varphi_s) $ is said to be a \textbf{subhistory} of a history $t = (\tau, \varphi)$ if $\tau_s = (V_s, E_s) $ consists of a node in $\tau $ and all its descendant nodes and edges, and $\varphi_s = \left . \varphi \right |_{V_s} $.
    % We say that $s $ is a \textbf{strict} subhistory of $t $ if $s\neq t $.
\end{defn}

However, we will primarily use a different definition in this text, which is equivalent up to isomorphism on internally labeled trees:

\begin{defn}\label{defn:isomorphism}
    Let $t = (V, E, \varphi) $ and $t' = (V', E', \varphi') $ be two labeled trees.
    Then $t $ and $t' $ are \textbf{isomorphic} if there exists a bijection $h:V\to V' $ which preserves labels and respects tree structure.
    That is,
    \begin{itemize}
        \item $\varphi(v) = \varphi'(h(v)) $ for all $v\in V $
        \item $E' = \left\{\left(h(v), h(v') \right) \st (v, v') \in E \right\} $
    \end{itemize}
\end{defn}

\begin{lemma}\label{lemma:correspondence_c_to_h}
    Let $(V, E) $ be the complete history sDAG on labels $Y $.
    Given a history $(V_t, E_t) $ in $(V, E) $, let
    \begin{equation*}
        V' = V_t\setminus \left\{\rho \right\},
    \end{equation*}
    and
    \begin{equation*}
        E' = E_t \setminus \left\{\left(\rho, v \right )  \right\},
    \end{equation*}
    with $v $ the only child node of the UA node in $(V_t, E_t) $.
    Define the function $\varphi : V'\to Y $ as $\varphi((\ell, U)) = \ell \in Y $.

    The correspondence $(V_t, E_t) \mapsto t_{(V_t, E_t)}=(V', E', \varphi) $ from histories in $(V, E) $ to labeled trees on labels $Y $ is well-defined.
\end{lemma}

\begin{proof}
    Given the history $(V_t, E_t) $, the label function restricted to leaf nodes, $\left . \varphi \right |_{L(t_{(V_t, E_t)})} : L(t_{(V_t, E_t)}) \to Y$ is an injection , since DAG leaf nodes are uniquely labeled by elements of $Y $.
    Also, any node $w $ in the labeled tree $t_{(V_t, E_t)} $ is determined by a unique node $v \in V_t $.
    $v $ must have either no child clades, or at least two child clades, and $v $ must have a child node for each child clade.
    Each child node of $v $ in $(V_t, E_t) $ corresponds to a child node of $w $ in $t_{(V_t, E_t)} $, so $w $ may not have exactly one child node.

    That is, the map named in the lemma is well-defined.
\end{proof}

\begin{lemma}\label{lemma:correspondence_h_to_c}
    Let $(V, E) $ be the complete history sDAG on labels $Y $.
    Let $t = (\tau, \varphi) $ be a labeled tree with labels in $Y$, and with root node $w_0$.

    For each node $w $ of $\tau $, let $C_w\subset Y$ be the set of leaf labels below the node $w $, and let
    \begin{equation*}
        v_w = \left(\varphi(w), \left\{C_{w'}\ \big|\ w'\text{ a child of } w \right\} \right).
    \end{equation*}
    Define
    \begin{equation*}
    V_t = \left\{v_w\ |\  w \text{ is a node of } t \right\} \cup \left\{\rho \right\},
    \end{equation*}
    and
    \begin{equation*}
        E_t = \left\{\left(v_{w_1}, v_{w_2} \right)\ |\ \left(w_1, w_2 \right)\text{ is an edge in } t \right\} \cup \left\{\left(\rho, v_{w_0} \right) \right\}.
    \end{equation*}
    The correspondence $t\mapsto (V_t, E_t) $ from labeled trees on labels $Y$ to histories in $(V, E) $ is well-defined.
\end{lemma}
\begin{proof}
\item \emph{The assignment $w\mapsto v_w = (\ell, U) $ of nodes in the labeled tree to nodes in the DAG is well-defined:}
    \begin{itemize}
        \item $U $ consists of disjoint subsets of $Y $ because $\varphi $ is injective on leaves of $t $, and sets of leaves between child nodes of $w $ must be disjoint.
        \item $U $ is either empty, or contains more than one subset of $Y $, since $w $ can be a leaf node with no children, or an interior node of $t $ with two or more children.
        \item $U=\emptyset $ if and only if $w $ is a leaf node, because $w $ has no children if and only if $w $ is a leaf node.
    \end{itemize}

\item \emph{This assignment $w\mapsto v_w $ is also injective:}
    In particular, no two nodes in a labeled tree may have the same subpartition.
    To see this, let $w_1, w_2$ be two nodes in a labeled tree $t $, with subpartitions $U_1 $ and $U_2 $.
    We will show that $U_1\neq U_2 $.

    If one of the nodes $w_1, w_2 $ is not reachable from the other, then $U_1\cap U_2 = \emptyset $, and $U_1\neq U_2 $.

    Otherwise, suppose that $w_2 $ is reachable from $w_1 $.
    $w_1 $ must have more than one child node, so $w_1 $ has at least one child node $w_1' $ from which $w_2 $ is not reachable.
    Therefore, the clade below $w_1' $ must be disjoint from all the clades in $U_2 $, since $t $ is a tree.
    However, the clade below $w_1' $ is an element of $U_1 $, so $U_1\neq U_2 $.\\
\item \emph{The assignment $(w_1, w_2)\mapsto (v_{w_1}, v_{w_2}) $ from edges in the labeled tree to edges in the history sDAG is well-defined and injective:}

        Either
    \begin{itemize}
    \item The union of the child clades of $w_2 $ are a child clade of $w_1 $ (in particular, the child clade under the child $w_2 $), or
    \item $w_2 $ is a leaf node and therefore one of the child clades of $w_1 $ is $\left\{\ell_{w_2} \right\} $.
    \end{itemize}

    The assignment on edges is also injective, since the assignment on nodes is injective.

\item \emph{$(V_t, E_t)$ is a history sDAG:}

    Since the assignments of nodes and edges in the labeled tree to nodes and edges in the complete DAG are well-defined, $V_t\subset V $ and $E_t\subset E $.
    To finish showing that $(V_t, E_t) $ is a history sDAG, notice that each node-clade pair $(v, C) $ has a descendant edge, namely the one to $v_w $, where $w $ is the parent node of the clade $C $ in $t $.
    Also, each node is reachable from the UA node, since each node in $t $ is reachable from $w_0 $.

    Notice $(V_t, E_t) $ has the same tree structure and labels as $t $, by construction.
    Therefore, such a choice of history $(V_t, E_t) $ is uniquely determined by a labeled tree $t $.
\end{proof}
\begin{lemma}\label{lemma:correspondence_bijection}
    \textbf{(Correspondence)} The map from histories to labeled trees named in \autoref{lemma:correspondence_c_to_h}, and the map from labeled trees to histories named in \autoref{lemma:correspondence_h_to_c}, are inverses, up to label-preserving bijection on nodes.
    In particular, both maps name a bijective correspondence between histories in a history sDAG $(V, E) $ on labels $Y $, and labeled trees with labels in $Y$.
\end{lemma}
\begin{proof}
    Given a labeled tree $t $, the labeled tree recovered from the history $(V_t, E_t) $ is exactly the tree we started with (up to the isomorphism in \autoref{defn:isomorphism}).

    In the other direction, the labeled tree $t_{(V_t, E_t)} $ derived from a history $(V_t, E_t) $, induces exactly the history $(V_t, E_t) $.
    This demonstrates the desired bijection.
\end{proof}

\subsection{Supplementary Tables}

\begin{table}[H]
    \centering
    \includegraphics{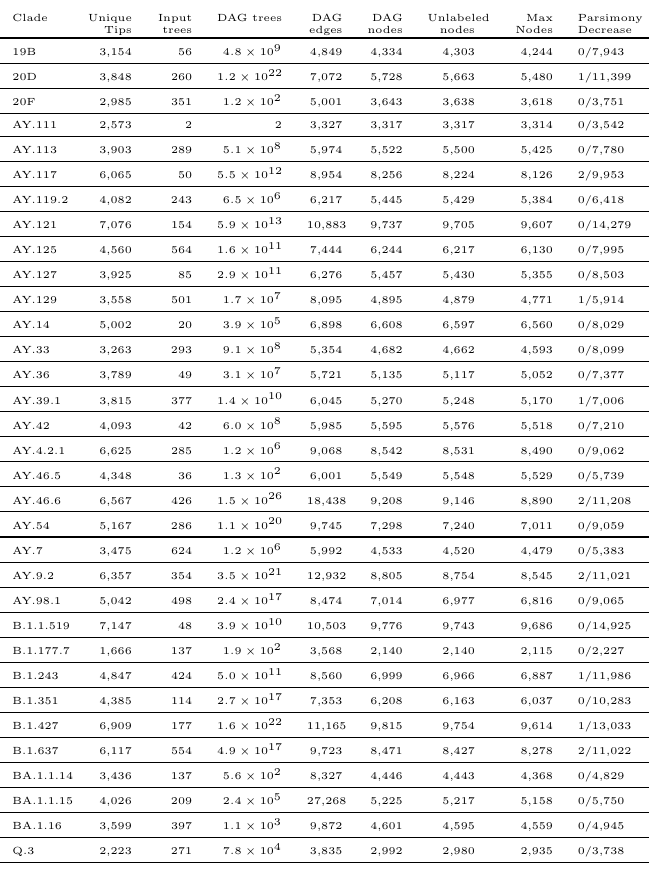}
\caption{\label{table:ushertrees}%
    Relevant characteristics of a history sDAG constructed on histories generated by \usher, for assorted clades defined in the public global \usher \sarscov tree.
    ``Max Nodes'' is the maximum number of nodes in any of the histories in the history sDAG\@.
    ``Parsimony Decrease'' is listed as $n / m $, where $n $ is the decrease in parsimony score in the history sDAG, relative to $m $, the minimum parsimony score of a history found by \usher.
}\end{table}

\end{document}